\documentclass[preprint,1p,times,12pt]{elsarticle}

\usepackage{multirow,pifont,array,caption,url,listings}

\usepackage[utf8]{inputenc} 

\usepackage{amsmath,amssymb,amsthm,extarrows}
\usepackage{lineno}
\usepackage[linesnumbered,lined,vlined,ruled,commentsnumbered]{algorithm2e}

\begin{document}

\begin{frontmatter}


\title{Sum-of-Max Partition under a Knapsack Constraint\tnoteref{label1}}
\tnotetext[label1]{Supported by National Natural Science Foundation of China 62002394.}

\author[sysu]{Kai Jin\corref{cor1}}
\cortext[cor1]{Corresponding author~ -- ~Email: \url{cscjjk@gmail.com}~ -- ~Homepage: \url{cscjjk.github.io}}

\affiliation[sysu]{organization={Sun Yat-Sen University},
             addressline={Gongchang Road 66},
             city={Shenzhen},
             postcode={518000},
             state={Guangdong},
             country={China}}

\author[sysu]{Danna Zhang}
\author[sysu]{Canhui Zhang}

\begin{abstract}
Sequence partition problems arise in many fields, such as sequential data analysis, information transmission, and parallel computing. In this paper, we study the following partition problem variant: given a sequence of $n$ items $1,\ldots,n$, where each item $i$ is associated with weight $w_i$ and another parameter $s_i$, partition the sequence into several consecutive subsequences, so that the total weight of each subsequence is no more than a threshold $w_0$, and the sum of the largest $s_i$ in each subsequence is minimized.

This problem admits a straightforward solution based on dynamic programming, which costs $O(n^2)$ time and can be improved to $O(n\log n)$ time easily. Our contribution is an $O(n)$ time algorithm, which is nontrivial yet easy to implement. We also study the corresponding tree partition problem. We prove that the problem on the tree is NP-complete and we present an $O(w_0 n^2)$ time ($O(w_0^2n^2)$ time, respectively) algorithm for the unit weight (integer weight, respectively) case.
\end{abstract}


\begin{keyword}
Sequence Partition \sep Tree Partition \sep Dynamic Programming Speed Up \sep Options Dividing Technique \sep Monotonic Queue.


\end{keyword}

\end{frontmatter}


\theoremstyle{plain}
\newtheorem{remark}{Remark}
\newtheorem{note}{Note}
\theoremstyle{theorem}
\newtheorem{theorem}{Theorem}
\newtheorem{corollary}{Corollary}
\newtheorem{fact}{Fact}
\newtheorem{lemma}{Lemma}
\newtheorem{definition}{Definition}
\newtheorem{problem}{Problem}

\section{Introduction}

Sequence and tree partition problems have been studied extensively since 1970s,
  due to their importance in parallel processing \cite{seq-2,app-parallel-1,app-parallel-2},
    task scheduling \cite{app-task-scheduling-1,app-task-scheduling-2},
      sequential data analysis \cite{app-data-analyis-1,app-data-analyis-2,app-data-analyis-3},
        network routing and telecommunication \cite{seq-1,app-routing-1,app-routing-2,app-routing-3}.
In this paper, we study the following partition problem variant:

\begin{description}
\item[Sequence partition] Given a sequence of $n$ items $1,\ldots,n$, where item $i$ is associated with a weight $w_i$ and a parameter $s_i$
  (which can be interpreted as the significance or safety level or distance from origin or CPU delaying time or length of object,
     of item $i$, depending on the different applications of the problem),
       partition the sequence into several consecutive subsequences,
         so that the total weight of each subsequence is no more than a given threshold $w_0$ (this will be referred to as the Knapsack constraint),
          and the objective is the sum of the largest $s_i$ in each subsequence, which should be minimized.
    Throughout, we assume that $w_1,\ldots,w_n,s_1,\ldots,s_n$ are nonnegative.

\item[Tree partition] Given a tree of $n$ nodes $1,\ldots,n$, where node $i$ is associated with a weight $w_i$ and a parameter $s_i$,
  partition the tree into several connected components,
    so that the total weight of each component is no more than $w_0$
       and the sum of the largest $s_i$ in each component is minimized.
\end{description}

For convenience, denote $W_{a,b}=\sum_{v:a\leq v \leq b} w_v$ and
                  $S_{a,b}=\max_v\{s_v \mid a\leq v\leq b\}$.
The sequence partition algorithm can be solved in $O(n^2)$ time by a straightforward dynamic programming
  of the following formulation: 
    $$F[i] = \min_{j<i}\{F[j]+S(j+1,i) \mid W(j+1,i)\leq w_0\}\quad (1\leq i\leq n).$$
Those $j$ appeared in this formula (satisfying $j<i$ and $W(j+1,i)\leq w_0$) are called the options of $i$, and $F[j]+S(j+1,i)$ is referred to as the value of $j$.
Organizing all these values by a min-heap, the running time can be improved to $O(n \log n)$.
The main contribution of this paper is a more satisfactory $O(n)$ time algorithm.

To obtain our linear time algorithm, we abandon the min-heap and use a more sophisticated data structure for organizing the candidate values.
We first show that computing $F[i]$ reduces to finding the best s-maximal option,
  where an option $j$ is \emph{s-maximal} if $s_j>S(j+1,i)$.
Interestingly, the s-maximal options fall into two categories:
    As $i$ grows, some of these options will be out of service due to the Knapsack constraint, and
          we call them \emph{patient options} --- they admit the first-in-first-out (FIFO) property clearly,
       whereas the other options will be out of service due to the s-maximal condition,
          and we call them \emph{impatient options} --- they somehow admit exactly the opposite property first-in-last-out (FILO).
We then use a monotonic queue \cite{book-IntroAlg} for organizing the values of patient options
   and a monotonic stack \cite{book-IntroAlg} for organizing the values of impatient options.
As a result, we find the best patient and impatient options, and thus the overall best option, in amortized $O(1)$ time,
        thus obtaining the linear time algorithm.
The difficulty lies in analyzing and putting the options into correct container --- the queue or the stack.
  Nontrivial mechanisms are applied for handling this; see section~\ref{sect:chain}.
In a final simplified version of our algorithm,
    we further replace the monotonic queue and stack by a deque, see a discussion in subsection~\ref{subsect:alg-final-on}.

Although our algorithm is inevitably more difficult to analyze compared to its alternative (based on heap),
  it is still quite easy to implement ---
    in fact, our implementation using \emph{C/C++ program} (given in \ref{sect:code}) contains only 30 lines.
The alternative algorithm is implemented as well for a comparison of the real performances.
    Experimental results show that our algorithm is much faster as $n$ grows large (60 times faster for $n=10^6$ in some cases); see \ref{sect:experiment}.

\smallskip Our second result says that the decision version of our tree partition problem (see problem~\ref{problem:tree-partition} in section~\ref{sect:tree}) is NP-complete.
For proving this result,
  we first show that a Knapsack problem variant (see problem~\ref{problem:knapsack-2} in section~\ref{sect:tree}) is NP-complete,
    and then prove that this Knapsack problem reduces to the tree partition problem,
       which proves that the latter is NP-complete.

In addition, we consider restricted cases of the tree partition problem where all the weights are (1) integers, or (2) unit.
  We show that the unit weight case admits an $O(w_0n^2)$ time solution (be aware that $w_0=O(n)$ in this case), whereas the integer weight case admits an $O(w_0^2n^2)$ time solution.
  The running time analysis of the unit weight case is based on a nontrivial observation (Lemma~\ref{lemma:C}).

\subsection{Motivations \& Applications}

Our partition problems are not only of theoretical value (because they have clean definitions),
  but also of practical value, as they can be applied in real-life.

\medskip In physical distribution, $n$ cargos with weights $w_1,\ldots,w_n$ in a center
  need to be loaded into vehicles and then be delivered to different destinations along a route,
    having distances $s_1,\ldots,s_n$ away from the center.
  Those cargos coming in a line but not exceeding a constraint $w_0$ can be loaded into the same vehicle.
A good partition of cargos is required for saving the total transportation fee.

Sometimes, cargos have the same destination but have different significance / fragile levels $s_1,\ldots,s_n$ and
  each vehicle buys an insurance according to the highest level of cargos it contains.
A good partition saves the total insurance fee.

In a more realistic situation,
   there are $k$ types of vehicles, each of different weight limit and rates on oil consumption,
     and we are allowed to select a vehicle for each batch of cargos.
We can model this by an extended partition problem and solve it in $O(kn)$ time (using the ideas for case $k=1$); see subsection~\ref{subsect:extension}.

\medskip Similar applications may be found in telecommunication / network routing,
  where we may want to send $n$ messages on time using the satellite or cable.
    The total length of message in each block is limited, which corresponds to the Knapsack constraint.
    Moreover, the higher safety level a message has, the more expensive communication channel we must use for sending it.
    Each block chooses a channel according to the highest safety level of the message it contains,
      and we want to partition the messages into blocks so that the total expense is minimized.

\medskip The partition problem finds applications in parallel computing and job scheduling.
  We may also interpret $s_1,\ldots,s_n$ as processing times of jobs.
  Each job requires some resources and the total resources a batch of jobs can apply is limited.

\subsection{Related work}

Sequence partition problems have been studied extensively in literature.
Olstad and Manne \cite{seq-1} presented an $O(k(n-k))$ time algorithm for
    finding a partition of a given sequence of length $n$ into $k$ pieces $\gamma_1,\ldots,\gamma_k$ so that $\max_i f(\gamma_i)$ is minimized, where $f$ is any prescribed, nonnegative, and monotone function.
P{\i}nar and Aykanat \cite{seq-2} designed an $O(k\log n+n)$ time algorithm for a special case of this problem
  where $f(\gamma_i)$ is defined as the sum of the weights of elements in $\gamma_i$.
  As a comparison, the problem studied in \cite{seq-2} aims to minimize the Max-of-Sum, whereas our problem aims to minimize the Sum-of-Max.
  Zobel and Dart \cite{seq-3} gave an $O(n)$ time algorithm for the following variant:
    Given a threshold value $L$, find a partition into $k$ pieces $\gamma_1,\ldots,\gamma_k$ so that the total weight of each piece $\gamma_i$ is at least $L$ and $\sum_i (\text{the weight of $\gamma_i$}-L)^2$ is minimized.

Tree partition is more complicated than sequence partition,
  and it has drawn more attention over the last four decades, especially in theoretical computer science.
    Given a threshold $w_0$ and a tree whose nodes have assigned weights,
       Kunda and Misra \cite{tree-1} showed a linear time algorithm for
         finding a partition of the tree into $k$ components (by deleting $k-1$ edges),
               so that each component has a total weight no more than $w_0$, meanwhile $k$ is minimized.
     Note that this problem is a special case of our tree partition problem (where $s_i$'s are set to be 1).
       Parley \textit{et. al} \cite{tree-2} considered partitioning a tree into the minimal number of components
               so that the diameter of each component is no more than a threshold $D_0$.
     Becker and Schach \cite{tree-3} gave an $O(Hn)$ time tree partition algorithm
        towards the minimal number of components so that
           the weight of each component is no more than a threshold $w_0$ and
              the height of each component is no more than another threshold $H$.
     Ito \textit{et. al} \cite{tree-4} partitioned a tree in $O(n^5)$ time
        into the minimum (or maximum, respectively) number of components with weights in a given range.

Pioneers in this area have also studied the tree partition problems
   in which the number of components $k$ is fixed
      and an objective function defined by the components is to be optimized.
   For example, maximize the minimum weight of the components \cite{tree-k-1},
      or minimize the maximum weight of components \cite{tree-k-2}.
   Surprisingly, both problems can be solved in linear time by parametric search; see Frederickson \cite{tree-k-3,tree-k-4}.
     Yet the linear time algorithm is extremely complicated.
     Agasi \textit{et. al} \cite{tree-k-5} showed that a variant of the min-max problem is NP-hard.

\newcommand{\cost}{\mathsf{cost}}
\newcommand{\weight}{\mathsf{weight}}
\newcommand{\nex}{\mathsf{next}}

\section{A linear time algorithm for the partition problem}\label{sect:chain}

The partition problem can be solved by dynamic programming as shown below.
Let $F[i]$ be the optimal value of the following optimization problem:

\begin{quote}
Partition $[1,i]$ into several intervals $I_1,\ldots,I_j$ such that
their total cost $\sum_{k=1}^{j} \cost(I_k)$ is minimized,
  subject to the constraint that
    the weight $\weight(I_k)$ of each interval $I_k~(1 \leq k\leq j)$ is less than or equal to $w_0$.
    Throughout, $\cost(I_k)=\max_{v \in I_k} s_v$ and $\weight(I_k)=\sum_{v\in I_k} w_v$.
\end{quote}

For convenience, denote $W_{a,b}=\sum_{v:a\leq v \leq b} w_v$ and
                  $S_{a,b}=\max_v\{s_v \mid a\leq v\leq b\}$.
The following transfer equation is obvious.
    \begin{equation}\label{eq:1}
	F[i]=\min_{j:0\leq j<i} \{ F[j]+ S_{j+1,i} \mid W_{j+1,i}\leq w_0\}.
    \end{equation}

Clearly, the partition problem reduces to computing $F[1],\ldots,F[n]$.
(Note: we should store the optimum decision $j$ of each $F[i]$ together with the cost $F[i]$. In this way the optimum partition with cost $F[i]$ can be easily traced back.)

Using formula \eqref{eq:1}, we can compute $F[1],\ldots,F[n]$ in $O(n^2)$ time. For computing $F[i]$, it takes $O(n)$ times
   to search the options of $i$ and select the best.

\subsection{An $O(n\log n)$ time algorithm using heap}\label{subsect:alg-heap-nlogn}

\newcommand{\Os}{\mathbb{O}}

To speed up the naive quadratic time algorithm above,
      we have to search the best option of each $i$ more efficiently.
        This subsection shows that we can find the best option in $O(\log n)$ time by utilizing a heap data structure.

Denote $O_i=\{j \mid 0\leq j <i, W_{j+1,i}  \leq W\}$ for each $i~(1\leq i \leq n)$.
       Call each element $j$ in $O_i$ an \emph{option} of $i$.
An option $j$ is called an \emph{s-maximal option} of $i$ if $j>0$ and $s_j> S_{j+1,i}$.
  Denote by $\Os_i$ the set of s-maximal options of $i$.

Denote $o_i=\min{O_i}$ and note that $O_i=[o_i, i-1]$.

\begin{lemma} \label{lem:1}
Set $\Os_i \cup \{o_i\}$ contains an optimal option of $F[i]$. As a corollary:
    \begin{equation}\label{eq:4}
    	F[i]=\min_j \left\{F[j]+S_{j+1,i} \mid j \in \Os_i \cup \{o_i\} \right\}.
    \end{equation}
\end{lemma}

\begin{proof}
Assume $j > o_i$ and $j$ is not s-maximal. As $j$ is not s-maximal, $s_j \leq S_{j+1,i}$, therefore (a) $S_{j,i} = S_{j+1,i}$.
Moreover, we have (b) $F[i-1] \leq F[i]$. The proof of this inequality is as follows.
Let $\Pi$ be the optimal partition of $1\ldots i$.
 Let $\Pi'$ be the same as $\Pi$ except for deleting $j$ (from the last interval).
Clearly, the cost of $\Pi'$ is at most the cost of $\Pi$ and the latter equals $F[i]$.
  Moreover, the cost of the best partition of $1\ldots i-1$ is no more than that of $\Pi'$.
    Together, $F[i-1] \leq F[i]$.
Combining (a) and (b), $F[j-1] + S_{j,i} =F[j-1] + S_{j+1,i} \leq F[j] +S_{j+1,i}$, which means option $j-1$ is no worse than $j$
  in computing $F[i]$.
    By the assumption of $j$, it follows that there is a best option of $F[i]$ that is s-maximal or equal to $o_i$.
\end{proof}

Without loss of generality, assume $\Os_i = \{j_1,\cdots,j_t\}$, where $j_1 < \cdots < j_t$. According to the  definition of s-maximal: $s_{j_1} > \cdots > s_{j_t}>s_i$.

We use a deque $J$ to store $\Os_i$ during the computation of $F[1],\ldots,F[n]$.
When we are about to compute $F[i]$, the deque $J$ shall be updated as follows:

\begin{enumerate}
    \item $i-1$ joins $J$ (to the tail).
    \item Several options $j$ at the tail of $J$ are popped out, since they do not satisfy the ``s-maximal constraint'' $s_j>s_i$.
    \item Several options $j$ at the head of $J$ are popped out, since they do not satisfy the ``weight constraint'' $W_{j+1,i} \leq w_0$
\end{enumerate}

Clearly, each $j~(1\leq j\leq n)$ will be pushed in and popped out from $J$ at most once,
  so the total time for maintaining $J$ in the algorithm is $O(n)$.
Below we show how to compute $F[1],\ldots,F[n]$ using $J$ (i.e., $\Os_i$) and the equation \eqref{eq:4}.

\begin{definition}\label{def:next(j)}
    For any s-maximal option $j$ (i.e., $j\in J$),
     let $\nex(j)$ be the option to the right of $j$ in the list $J=\{j_1,\ldots,j_t\}$; and define $\nex(j_t) = i$.
     Note that $\nex(j)$ is variant while $i$ increases and note that
     $S_{j+1,i}=s_{\nex(j)}$. For convenience, denote $$\cost[j] = F[j] +s_{\nex(j)}$$

    Furthermore, let $j^* = \mathop{\arg\min}_{j \in J} \cost[j]$. To be precise, if $J=\varnothing$, define $j^*=-1$.
    Let $u = \mathop{\arg\max}_{o_i < j \leq i}\ s_j$ (if not unique, let $u$ be the largest one of them).
\end{definition}

Obviously, $u=\left\{
  \begin{array}{ll}
    \nex(o_i), & o_i \in  \Os_i; \\
    \min{\{ \Os_i \cap \{i\}\}}, & o_i \notin  \Os_i.
  \end{array}
\right.$
  (by the monotonicity of $J$).

\smallskip Equipped with these notations, equation \eqref{eq:4} can be simplified as follows:

\begin{equation} \label{eq:5}
	F[i]=
	\begin{cases}
		\min (F[o_i]+s_u, \cost[j^*]) &j^*\neq -1\\
		F[o_i]+s_u & j^*=-1
	\end{cases}
\end{equation}

\begin{proof}
    When $j^*\neq -1$, set $J$ is not empty, and we have
     \begin{equation*}
	\begin{aligned}
		F[i]&=\min\left(F[o_i] + S_{o_i+1,j}, \min_{j\in J}\{F[j]+S_{j+1,i}\}\right)~\quad \text{(according to \eqref{eq:4})}\\
		&=\min \left(F[o_i]+s_u, \min_{j\in J}\{F[j]+s_{\nex(j)}\}\right)\\
		&=\min \left(F[o_i]+ s_u, \min_{j\in J} \cost[j] \right)=\min (F[o_i]+ s_u, \cost[j^*])
	\end{aligned}
    \end{equation*}
    When $j^* = -1$, set $J=\Os_i=\varnothing$ and $F[i]=F[o_i] + S_{o_i+1, i} = F[o_i] + s_u$.
\end{proof}

We can compute $F[1],\ldots,F[n]$ in $O(n\log n)$ time based on formula \eqref{eq:5}.
Notice that $o_i$ can be computed in $O(1)$ amortized time, and so as $u$, which can be computed easily from $J$.
The challenge only lies in computing $j^*$ and $\cost[j^*]$.

For computing $j^*$ and $\cost[j^*]$ efficiently,
  we organize $\{(\cost[j], j)\mid j \in J\}$ into a min-heap. Then, $j^*$ can be found in O(1) time.
Note that $\cost[j]$ changes only if $\nex[j]$ changes, and moreover, at most one value in the $\nex$ array changes when $i$ increases by 1.
It follows that elements in $\{(\cost_j, j) \mid j \in J\}$ would change at most $O(n)$ times during the process of the algorithm.
Further since each change takes $O(\log n)$ time, the total running time is $O(n \log n)$ time.

\subsection{An $O(n)$ time algorithm using a novel grouping technique}
This section shows a novel \emph{grouping technique} that computes $j^*$ in $O(1)$ time.
For describing it, a concept called ``renew'' needs to be introduced.

\begin{definition} \label{def:renew}
 We say an s-maximal option $j$ is \emph{renewed} when $\nex(j)$ changes. An option $j$ is regarded as a new option after being renewed, which is different from the previous $j$ --- the same $j$ with different $\nex(j)$ will be treated differently.
 \end{definition}

With this concept, the way for an option $j$ to exit $J$ falls into three classes:

\begin{enumerate}
\item (as $i$ increases) $j$ pops out from the head of the deque, since the constraint $W_{j+1, i} \leq w_0$ is no longer satisfied.

\item (as $i$ increases) $j$ pops out from the tail of the deque, since the constraint $s_j>s_i$ is not satisfied.

\item (as $i$ increases) $j$ is renewed; the old $j$ pops out and a new $j$ is added to $J$.
\end{enumerate}

\textbf{Note.}
  1. Assume that the weight constraint $W_{j+1,i} \leq w_0$ is checked before the s-maximal constraint $s_j > s_i$.
      That is, if an option violates both constraints, we regard that it pops out in the first way.
2. In each iteration, after some options pop out in the second way,
  the last option $j$ in $J$ (if $J\neq \varnothing$) will be renewed.

\medskip We divide the options into two groups: the patient ones and impatient ones.

\begin{definition} \label{def:patien and impatient}
An option that exits $J$ in the first way is called a \emph{patient option}.
An option that remains in $J$ until the end of the algorithm is also called a \emph{patient option}.
Other options are called \emph{impatient options}.
\end{definition}

\begin{figure}[h]
  \centering
  \includegraphics[width=.7\textwidth]{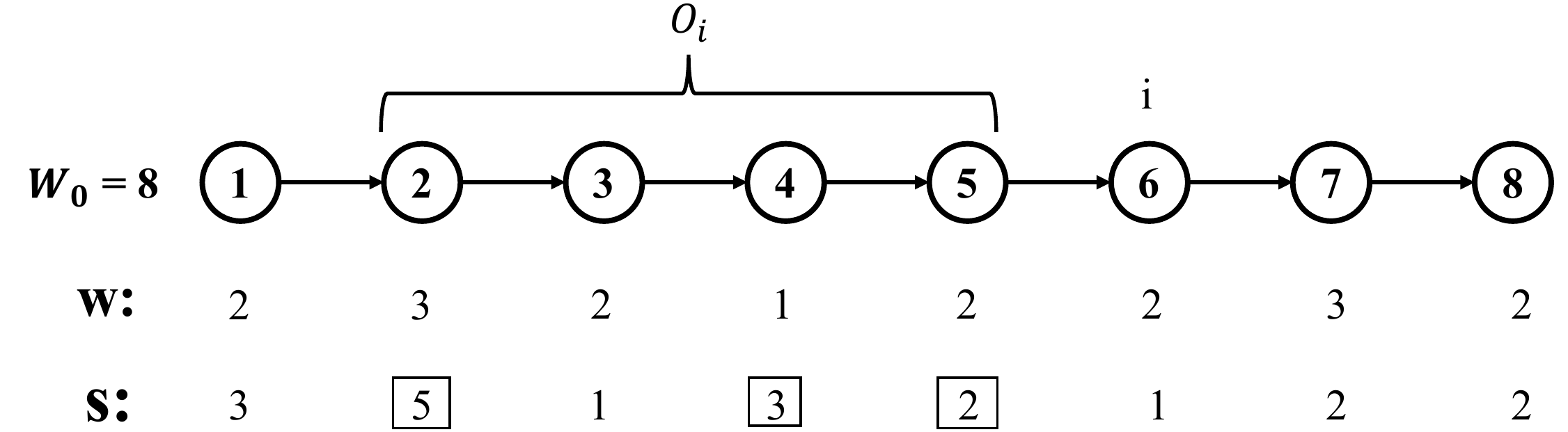}\\
  \caption{Patient and impatient options. Suppose $i=6$. The s-maximal options are 2, 4, 5.
    Option 2 would exit $J$ in the first way when $i=7$. So, it is patient.
    Options 4 and 5 are impatient.
    When $i=7$, option~5 would exit $J$ in the second way and option~4 in the third way.
    Note in particular that after option~4 is renewed, it becomes patient.}\label{fig:1}
\end{figure}

See Figure~\ref{fig:1} for an illustration of patient and impatient options.
   As can be seen from this illustration: An option $j$ may belong to different groups before and after renew,
      such as $j=4$ in the example.
   Because of this, the options before and after renew must be distinguished so that each option has its own group.

\medskip Denote the set of patient options by $\text{$J$}^{(p)}$ and the set of impatient options by $\text{$J$}^{(ip)}$.
  Obviously, $\text{$J$} = \text{$J$}^{(p)} \cup \text{$J$}^{(ip)}$. The idea of our algorithm is briefly as follows:
        First, find the best option in $\text{$J$}^{(p)}$ and the best option in $\text{$J$}^{(ip)}$.
            Then, choose the better one between them to be $j^*$.
    Two subproblems are yet to be resolved:

1. How to determine the group a newly added or renewed option belongs to?

2. How to efficiently obtain the optimal option in $\text{$J$}^{(p)}$ and $\text{$J$}^{(ip)}$ respectively?

Towards a linear time algorithm, we should resolve them in constant time.

\subsubsection{Determine whether an option is patient or impatient}

We associate each option $j~(1 \leq j \leq n)$ with a counter, denoted by $counter[j]$,
   which stores the number of times that $j$ would exit in the second or third way in the future.
For $j\in J$, we determine that $j$ is patient if and only if $counter[j]=0$.

In the following, we present a preprocessing algorithm (see Algorithm \ref{algo:2}) that obtains the counters at the initial state.
In the main process, when an option is to be renewed, we decrease its corresponding counter by 1;
 and if $counter[j]$ drops to 0 at that point, we get that option $j$ becomes patient from impatient.

\begin{algorithm}[h]
\LinesNumbered
\SetAlgoNoLine
\begin{small}
$o \leftarrow 0$\;
\For{$i = 1$ \textbf{to} $n$}{
    \lWhile{ $W[o+1,i] > w_0 $}{
        $o++$
    }
    \lWhile{$J~\land~W(J.head+1, i) > w_0$}{
        $J$.deleteHead()
    }
    \lWhile{$J~\land~s[J.tail] \leq s[i] $}{
        $counter[J.tail]++$, $J$.deleteTail()
    }
    \lIf{\text{$J$}}{
        $counter[J.tail]++$
    }
    $J$.insertTail($i$)\;   
    \lIf{$J.head = o$}{
        $u[i] \leftarrow J.second$
    }
    \lElse{
      $u[i] \leftarrow J.head$
    }
}
\end{small}
\caption{preprocess}
\label{algo:2}
\end{algorithm}

The preprocessing algorithm simulates the change of $J$ in advance.

Line 4-5 in Algorithm \ref{algo:2}: Deal with the options that exit $J$ by the first and second ways. Within the second way, the corresponding counter increases by 1.

Line 6 in Algorithm \ref{algo:2}: $J$.tail is renewed, thus $counter[J.tail]$++.

Line 8-9 in Algorithm \ref{algo:2}: Compute the value of $u$ for option $i$. Recall variable $u$ in Definition~\ref{def:next(j)} and \eqref{eq:5}. (Note: It would be troublesome to compute $u$ until the main process, since the main process no longer maintains $J$ as we will see.)

\smallskip Algorithm~\ref{algo:2} runs in $O(n)$ time. The analysis is trivial and omitted.

\newcommand{\Jp}{J^{(p)}}
\newcommand{\Jip}{J^{(ip)}}

\subsubsection{Compute the optimal option in $\text{$J$}^{(p)}$ and $\text{$J$}^{(ip)}$}

The following (trivial) observations are crucial to our algorithm.
\begin{enumerate}
\item When an option exit $\Jp$, it must be the smallest one in $\Jp$.
    In other words, the options in $\Jp$ (i.e. patient options) are \textbf{first-in-first-out (FIFO)}.
\item When an option exit $\Jip$, it must be the largest one in $\Jip$.
    In other words, the options in $\Jip$ (i.e. impatient options) are \textbf{first-in-last-out (FILO)}.
\end{enumerate}

Indeed, the options in $J$ are partitioned carefully into two groups (i.e. patient / impatient) such that they are either FIFO or FILO in each group. By doing this, the best option in each group might be found efficiently as shown below.

\smallskip We use a queue and a stack to store $\Jp,\Jip$, respectively.
  The maintenance of $\Jp,\Jip$ are similar to that of $J$, which are summarized in the following.

\begin{enumerate}
    \item Before computing $F[i]$, if $s_{i-1}>s_i$, the s-maximal option $i-1$ needs to be added into $\Jp$ or $\Jip$,
       depending on whether $counter[i-1]=0$ or not.
    \item Some options at the head of queue $\Jp$ are popped out, since they no longer satisfy the constraint ``$W_{j+1,i} \leq w_0$'', and some options at the top of stack $\Jip$ are popped out, since they do not satisfy the constraint ``$s_j>s_i$''.
    \item If $\Jip\neq \varnothing$ after step 2, the counter of $j=\Jip.top$ is decreased by 1, meanwhile $\nex(\Jip.top)$ becomes $i$. If $counter[j]$ drops to $0$, option $j$ becomes patient, and
          we transfer $j$ to $\Jp$ from $\Jip$ accordingly.
\end{enumerate}

\textbf{Note 1.} An option in $\Jp$ can leave only due to the weight constraint $W_{j+1,i} \leq w_0$,
  so it is unnecessary to check whether the tail of $\Jp$ satisfies $s_j > s_i$.
    Likewise, it is unnecessary to check the weight constraints of options in $\Jip$.

\textbf{Note 2.} When an option $j$ is transferred to $\Jp$ from $\Jip$, it can be added to the tail of queue $\Jp$ in $O(1)$ time.
  At this time, $j$ is renewed, which means that it is the largest option in $J$. Hence it can be directly added to the tail of $\Jp$.

Throughout, the options in $\Jp$ and $\Jip$ are in ascending order from head to tail, or bottom to top.
Each option joins and exits $\Jp$ and $\Jip$ at most once respectively.
  Therefore the maintenance of $\Jp,\Jip$ takes $O(1)$ amortized time.

\medskip Next, we show how to quickly compute the optimal options in $\Jp$ and $\Jip$
 respectively. To this end, we use the monotonic queue and monotonic stack.

First, we define the concept called \emph{dead}.

\begin{definition} \label{def:dead}
 Consider any option $j\in \Jp$ ($j\in \Jip$, respectively). If there is another option $j'$ in $\Jp$ ($\Jip$, respectively) with $\cost[j'] \leq \cost[j]$ and that $j'$ stays in $\Jp$ ($\Jip$, respectively) as long as $j$ does, then $j$ is regarded \emph{dead}.
(Note: In this definition, the renewed option is still regarded as a different option.)
\end{definition}

\begin{lemma} \label{lem:2}~

(1) Suppose $j,j'\in \Jp$. If $j<j'$ and $\cost[j'] \leq \cost[j]$, option $j$ is dead;

(2) Suppose $j,j'\in \Jip$. If $j'<j$ and $\cost[j'] \leq \cost[j]$, option $j$ is dead.
\end{lemma}

\begin{proof}
First, we prove (1). Because $j<j'$, we know $j$ is closer to the head than $j'$ in the queue,
  which means $j'$ leaves $\Jp$ later than $j$. By definition \ref{def:dead}, $j$ is dead.
Next, we prove (2). Because $j'<j$, we know $j$ is closer to the top than $j'$ in the stack,
  which means $j'$ leaves $\Jip$ later than $j$. By definition \ref{def:dead}, $j$ is dead.
\end{proof}

\newcommand{\Kp}{K^{(p)}}
\newcommand{\Kip}{K^{(ip)}}

To compute the optimal option of $\Jp$ or $\Jip$, we only need to focus on the options that are not dead.
The dead ones are certainly not optimal by definition.
(To be rigorous, there is always an optimal option that is not dead.)

Denote by $\text{K}^{(p)} =(p_1, \cdots, p_a)$ all the patient options that are not dead.

Denote by $\text{K}^{(ip)} =(q_1, \cdots, q_b)$ all the impatient options that are not dead.

Assume that $p_1 < \cdots <p_a$ and $q_1 < \cdots <q_b$.
As a corollary of Lemma \ref{lem:2}, $\cost[p_1] < \cdots < \cost[p_a]$, whereas $\cost[q_1] > \cdots > \cost[q_b]$.
Therefore, the optimal option in $\Jp$ is $p_1$ and the optimal option in $\Jip$ is $q_b$.

It remains to explain how to maintain $\Kp$ and $\Kip$ in $O(1)$ amortized time.
Because $\Kp$ is a monotonic subsequence of $\Jp$ and $\Kip$ is a monotonic subsequence of $\Jip$,
  the maintenance of $\Kp,\Kip$ resemble that of $\Jp,\Jip$. Details are summarized below.
  (Note: the cost of option $j$ is always stored in $\cost[j]$).

\begin{enumerate}
    \item After adding an option to the tail of $\Kp$, if $\cost[p_a] \leq \cost[p_{a-1}]$, then $p_{a-1}$ is dead, and hence it would be removed from deque $\Kp$. Repeat this until $\cost[p_a] > \cost[p_{a-1}]$. Zero or multiple options in $\Kp$ are deleted.
    \item After adding an option to the top of $\Kip$, if $\cost[q_b] \geq \cost[ q_{b-1}]$, then $q_b $ is dead, and it would be popped out of the stack directly. Otherwise, we have $\cost[q_1] > \ldots > \cost[q_b]$, and $q_b$ remains in the stack.
    \item When we want to delete some options from $\Kp$ or $\Kip$ (due to the weight or s-maximal condition),
       no additional operation is required except the deletion itself.
\end{enumerate}

\subsection{Combine $\Kp$ and $\Kip$ to simplify the above $O(n)$ algorithm}\label{subsect:alg-final-on}

The $O(n)$ time algorithm shown in the last subsection applies two data structures $\Kp$ and $\Kip$, which are monotonic queue or stack.
This subsection simplifies the algorithm by combining the two data structures into a deque.

First, we state a relationship between patient and impatient options.

\begin{figure}[h]
  \centering
  \includegraphics[width=.7\textwidth]{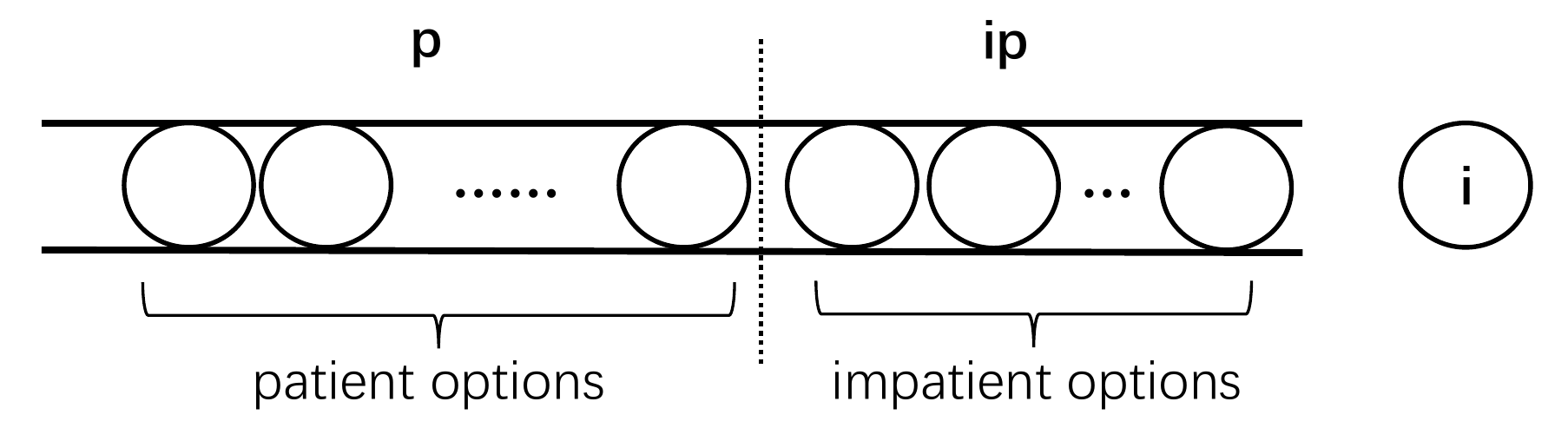}\\
  \caption{The distribution of patient and impatient options.}\label{fig:2}
\end{figure}

\begin{lemma} \label{lem:3}
The patient options are less than the impatient options in $J=\Os_i$.
\end{lemma}

\begin{proof}
 Take any impatient option $j$. Since $j$ is impatient,
   it will leave $J$ by the second or third way,
      so $j$ is at the tail of $J$ when it is removed.
 This means that the options to the right of $j$ must leave $J$ at its tail as well
  (they cannot leave at the head of $J$ since $j$ is over there, in front of them).
Therefore, the options to the right of $j$ must be impatient, which implies the lemma. See Figure~\ref{fig:2}.
\end{proof}

Recall that $\Kp$ and $\Kip$ consist of options that are not dead and $\Kp\subseteq \Jp$ and $\Kip\subseteq \Jip$.
As a corollary of Lemma \ref{lem:3}, $\Kp$ are to the left of $\Kip$.

\begin{figure}[h]
  \centering
  \includegraphics[width=.7\textwidth]{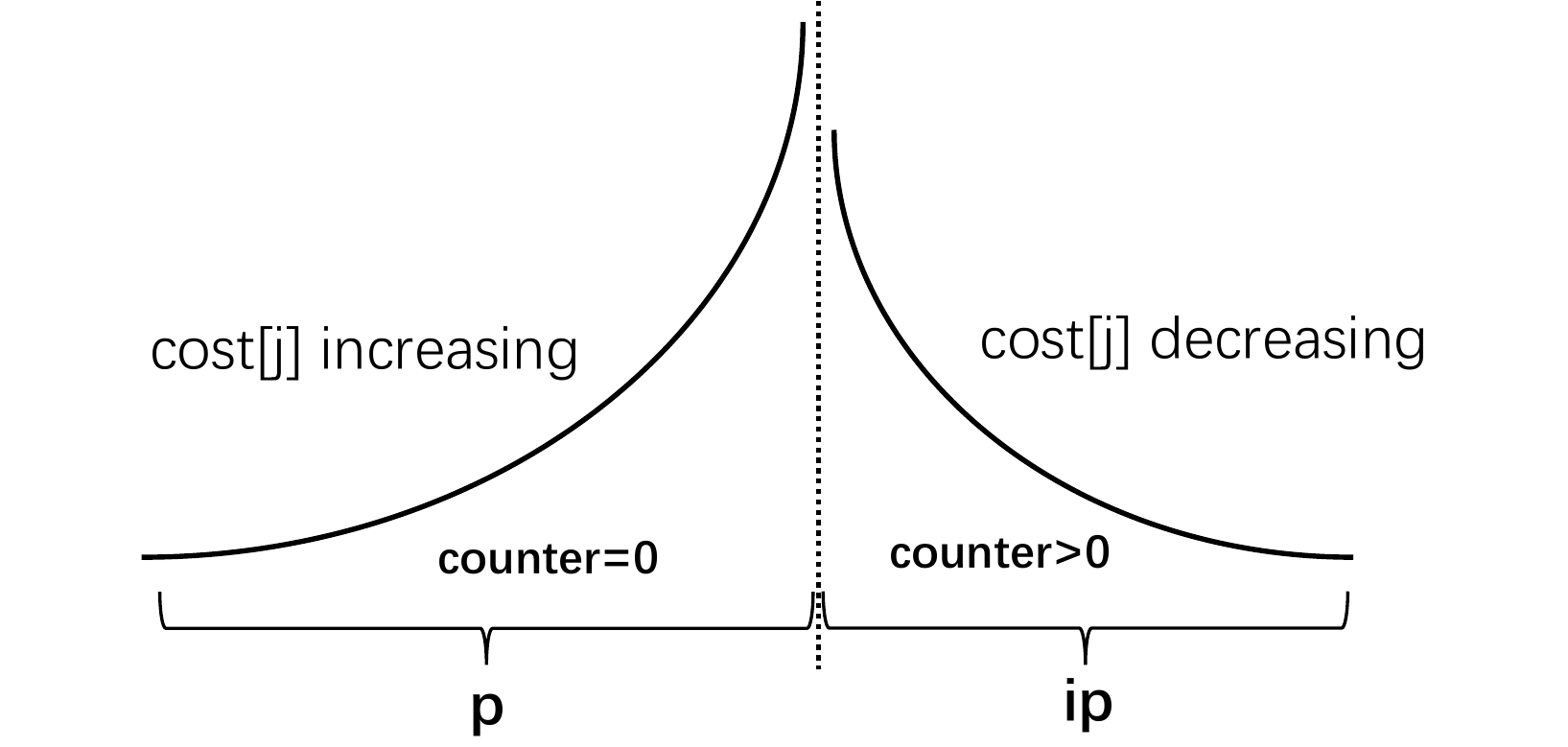}\\
  \caption{The cost distribution of options in $\Kp$ an $\Kip$}\label{fig:3}
\end{figure}

Our final algorithm replaces $\Kp$ and $\Kip$ by a deque $K$, whose left part (head) is $\Kp$ ($counter=0$) and the right part (tail) is $\Kip$ ($counter>0$).

The costs of options in the head (i.e. $\Kp$) is monotonically increasing,
  and the costs of options in the tail (i.e. $\Kip$ ) is monotonically decreasing, as shown in Figure~\ref{fig:3}.
    In particular, the optimal option in $K$ is at the head or tail of $K$.

\medskip The maintenance of $K$ is similar to the maintenance of $\Kp$ and $\Kip$ separately. Algorithm~\ref{algo:3} demonstrates the process for maintaining $K$ and computing $F[1],\ldots,F[n]$.
Recall the preprocessing algorithm in Algorithm~\ref{algo:2}.\smallskip

\begin{algorithm}[H]
\LinesNumbered
\SetAlgoNoLine
\begin{small}
$o \leftarrow 0$\;
\For{$i = 1$ \textbf{to} $n$}{
    \lWhile{$K~\land~W(K.head+1, i) > w_0$}{
        $K$.deleteHead()
    }
    \lWhile{$K~\land~s[K.tail] \leq s[i] $}{
        $K$.deleteTail()
    }
    \If{$K~\land~\cost[K.tail]\leq F[K.tail]+s[i]$}{
        $\cost[K.tail] \leftarrow F[K.tail] + s[i]$; \quad $counter[K.tail]--$\;
    }
    \lIf{$|K|>1~\land~counter[K.tail_2] > 0~\land~ \cost[K.tail_2] \leq \cost[K.tail]$}{
        $K$.deleteTail(); \quad // ``$K.tail_2$'' refers to the second last option in $K$
    }

    \lWhile{$|K|>1~\land~counter[K.tail]=0~\land~\cost[K.tail_2] \geq \cost[K.tail]$}{
        $K$.deleteTail2(); \quad // ``deleteTail2'' = delete the second last option in $K$
    }
    \lWhile{$W[o+1, i] > w_0$}{
        $o++$
    }
    $F[i] \leftarrow F[o] + s[u[i]]$\;
    \lIf{$K$}{
        $F[i] \leftarrow \min \{ \cost[K.head], \cost[K.tail], F[i]\}$
    }
    K.insertTail($i$); \quad $\cost[i] \leftarrow -1$\;
}
\end{small}
\caption{compute {$F[i]$}}
\label{algo:3}
\end{algorithm}

\smallskip Line~3 in Algorithm \ref{algo:3}: Some options at the head of $K$ exit by the first way.

Line~4 in Algorithm \ref{algo:3}: Some options at the tail of $K$ exit by the second way.

\smallskip Lines~5-7 in Algorithm \ref{algo:3}: After Line~4, the largest s-maximal option $J.tail$
        shall be renewed as $\nex(J.tail)$ becomes $i$.
        But be aware that $J.tail$ could be dead and if so, we need to do nothing.
        Observe that $J.tail$ is not dead if and only if $J.tail=K.tail$.
        Moreover, $J.tail=K.tail$ occurs if and only if $\cost[K.tail] \leq F[K.tail]+s[i]$.
        When the last condition holds (as checked by Line~5), we renew $K.tail$ at Line~6.
        (This avoids computing $J.tail$ and comparing it to $K.tail$).

\smallskip Lines~8-9 in Algorithm \ref{algo:3}: Remove the dead options.
Because a new option (including the renewing one) can join $K$ only at its tail,
  we can find dead options through comparing $K.tail_2$ and $K.tail$ as follows.
If $counter[K.tail_2] > 0$, the last two options of $K$ belong to $\Kip$.
 In this case, if $\cost[K.tail] \geq \cost[K.tail_2]$, $K.tail$ is dead and thus deleted.
When $counter[K.tail] = 0$, the last two options in $K$ belong to $\Kp$.
 We then check if $\cost[K.tail_2] \geq \cost[K.tail]$. If so, $K.tail_2$ is dead and thus deleted.
  Repeat it as long as $\cost[K.tail_2] \geq \cost[K.tail]$.

\medskip Figure~\ref{fig:4} shows an example where $n=8$.

\begin{figure}[h]
  \centering
  \includegraphics[width=.75\textwidth]{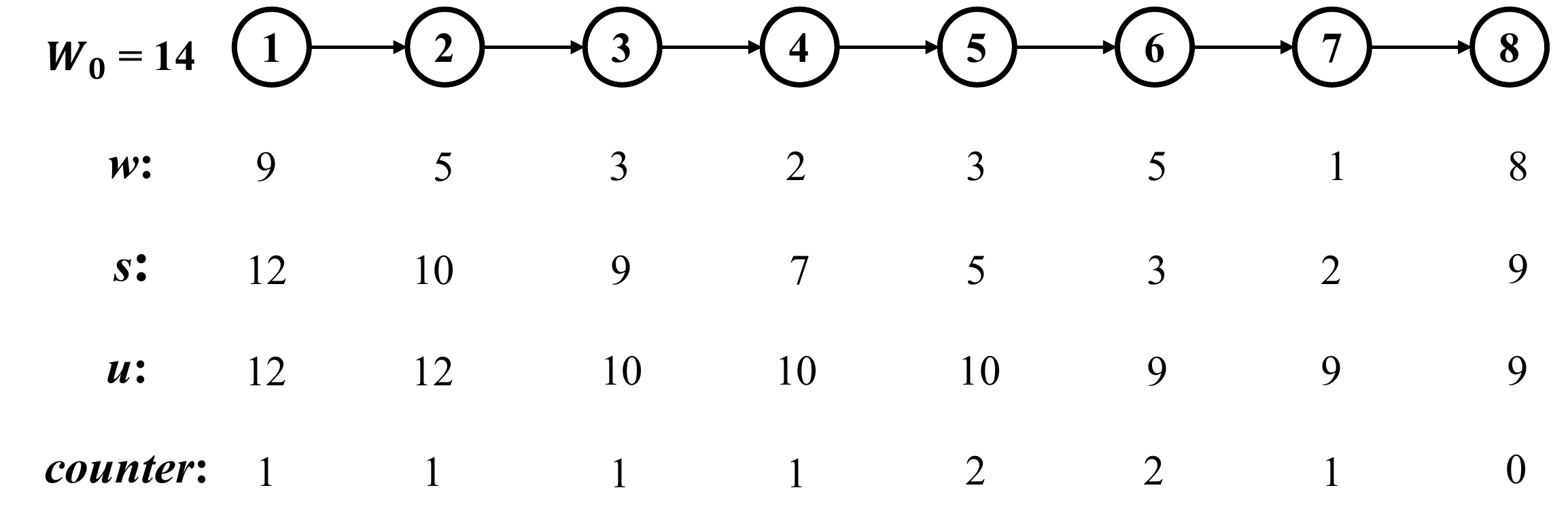}\\
  \caption{Illustration of an example.}\label{fig:4}
\end{figure}

We simulate the whole computation process for the example above and
   the deque $K$ at each iteration of $i$ is shown in Table~\ref{table:1}.

\begin{table}[h!]
    \centering
\begin{scriptsize}
    \begin{tabular}{cccccc}
    \hline
    \multicolumn{1}{|c|}{$\boldsymbol{i}$}    & \multicolumn{2}{c|}{$K$}                                                                                                                                                                                                 & \multicolumn{1}{c|}{$\boldsymbol{o_{i}}$} & \multicolumn{1}{c|}{$\boldsymbol{F[i]}$}  \\ \hline
    \multicolumn{1}{|c|}{}                    & \multicolumn{1}{c|}{}                         & \multicolumn{1}{l|}{NULL }                                      & \multicolumn{1}{c|}{}                     & \multicolumn{1}{c|}{}                     \\ \cline{2-3}
    \multicolumn{1}{|c|}{}                    & \multicolumn{1}{c|}{\textit{cost}}    & \multicolumn{1}{l|}{}                                                                                                                                                                                                 & \multicolumn{1}{c|}{}                     & \multicolumn{1}{c|}{}                     \\ \cline{2-3}
    \multicolumn{1}{|c|}{\multirow{-3}{*}{1}} & \multicolumn{1}{c|}{\textit{counter}} & \multicolumn{1}{l|}{}                                                                                                                                                                                                 & \multicolumn{1}{c|}{\multirow{-3}{*}{0}}  & \multicolumn{1}{c|}{\multirow{-3}{*}{12}} \\ \hline
    \multicolumn{1}{|c|}{}                    & \multicolumn{1}{c|}{}                                  & \multicolumn{1}{l|}{\,\footnotesize{\ding{172}}}                                                                                                                                              & \multicolumn{1}{c|}{}                     & \multicolumn{1}{c|}{}                     \\ \cline{2-3}
    \multicolumn{1}{|c|}{}                    & \multicolumn{1}{c|}{\textit{cost}}    & \multicolumn{1}{l|}{22}                                                                                                                                                                                               & \multicolumn{1}{c|}{}                     & \multicolumn{1}{c|}{}                     \\ \cline{2-3}
    \multicolumn{1}{|c|}{\multirow{-3}{*}{2}} & \multicolumn{1}{c|}{\textit{counter}} & \multicolumn{1}{l|}{~0}                                                                                                                                                                                                & \multicolumn{1}{c|}{\multirow{-3}{*}{0}}  & \multicolumn{1}{c|}{\multirow{-3}{*}{12}} \\ \hline
    \multicolumn{1}{|c|}{}                    & \multicolumn{1}{c|}{}                                  & \multicolumn{1}{l|}{\,\footnotesize{\ding{173}}}                                                                                                                                              & \multicolumn{1}{c|}{}                     & \multicolumn{1}{c|}{}                     \\ \cline{2-3}
    \multicolumn{1}{|c|}{}                    & \multicolumn{1}{c|}{\textit{cost}}    & \multicolumn{1}{l|}{21}                                                                                                                                                                                               & \multicolumn{1}{c|}{}                     & \multicolumn{1}{c|}{}                     \\ \cline{2-3}
    \multicolumn{1}{|c|}{\multirow{-3}{*}{3}} & \multicolumn{1}{c|}{\textit{counter}} & \multicolumn{1}{l|}{~0}                                                                                                                                                                                                & \multicolumn{1}{c|}{\multirow{-3}{*}{1}}  & \multicolumn{1}{c|}{\multirow{-3}{*}{21}} \\ \hline
    \multicolumn{1}{|c|}{}                    & \multicolumn{1}{c|}{}                                  & \multicolumn{1}{l|}{\,\footnotesize{\ding{173}}\qquad \hspace{2pt}\footnotesize{\ding{174}}}                                                                                               & \multicolumn{1}{c|}{}                     & \multicolumn{1}{c|}{}                     \\ \cline{2-3}
    \multicolumn{1}{|c|}{}                    & \multicolumn{1}{c|}{\textit{cost}}    & \multicolumn{1}{l|}{21\qquad\quad 28}                                                                                                                                                                        & \multicolumn{1}{c|}{}                     & \multicolumn{1}{c|}{}                     \\ \cline{2-3}
    \multicolumn{1}{|c|}{\multirow{-3}{*}{4}} & \multicolumn{1}{c|}{\textit{counter}} & \multicolumn{1}{l|}{~0\qquad\quad ~~0}                                                                                                                                                                          & \multicolumn{1}{c|}{\multirow{-3}{*}{1}}  & \multicolumn{1}{c|}{\multirow{-3}{*}{21}} \\ \hline
    \multicolumn{1}{|c|}{}                    & \multicolumn{1}{c|}{}                                  & \multicolumn{1}{l|}{\,\footnotesize{\ding{173}}\qquad\hspace{3pt}\footnotesize{\ding{175}}}                                                                                               & \multicolumn{1}{c|}{}                     & \multicolumn{1}{c|}{}                     \\ \cline{2-3}
    \multicolumn{1}{|c|}{}                    & \multicolumn{1}{c|}{\textit{cost}}    & \multicolumn{1}{l|}{21\qquad\quad 26}                                                                                                                                                                        & \multicolumn{1}{c|}{}                     & \multicolumn{1}{c|}{}                     \\ \cline{2-3}
    \multicolumn{1}{|c|}{\multirow{-3}{*}{5}} & \multicolumn{1}{c|}{\textit{counter}} & \multicolumn{1}{l|}{~0\qquad\quad ~~0}                                                                                                                                                                          & \multicolumn{1}{c|}{\multirow{-3}{*}{1}}  & \multicolumn{1}{c|}{\multirow{-3}{*}{21}} \\ \hline
    \multicolumn{1}{|c|}{}                    & \multicolumn{1}{c|}{}                                  & \multicolumn{1}{l|}{\,\footnotesize{\ding{173}}\qquad\hspace{3pt}\footnotesize{\ding{175}}\qquad\hspace{3pt}\footnotesize{\ding{176}}}                                                & \multicolumn{1}{c|}{}                     & \multicolumn{1}{c|}{}                     \\ \cline{2-3}
    \multicolumn{1}{|c|}{}                    & \multicolumn{1}{c|}{\textit{cost}}    & \multicolumn{1}{l|}{21\qquad\quad 26\qquad\quad 24}                                                                                                                                                 & \multicolumn{1}{c|}{}                     & \multicolumn{1}{c|}{}                     \\ \cline{2-3}
    \multicolumn{1}{|c|}{\multirow{-3}{*}{6}} & \multicolumn{1}{c|}{\textit{counter}} & \multicolumn{1}{l|}{~0\qquad\quad ~~0\qquad\quad ~1}                                                                                                                                                    & \multicolumn{1}{c|}{\multirow{-3}{*}{2}}  & \multicolumn{1}{c|}{\multirow{-3}{*}{21}} \\ \hline
    \multicolumn{1}{|c|}{}                    & \multicolumn{1}{c|}{}                                  & \multicolumn{1}{l|}{\,\footnotesize{\ding{173}}\qquad\hspace{3pt}\footnotesize{\ding{175}}\qquad\hspace{3pt}\footnotesize{\ding{176}}
  \qquad\hspace{3pt}\footnotesize{\ding{177}}} & \multicolumn{1}{c|}{}                     & \multicolumn{1}{c|}{}                     \\ \cline{2-3}
    \multicolumn{1}{|c|}{}                    & \multicolumn{1}{c|}{\textit{cost}}    & \multicolumn{1}{l|}{21\qquad\quad 26\qquad\quad 24\qquad\quad 23}                                                                                                                          & \multicolumn{1}{c|}{}                     & \multicolumn{1}{c|}{}                     \\ \cline{2-3}
    \multicolumn{1}{|c|}{\multirow{-3}{*}{7}} & \multicolumn{1}{c|}{\textit{counter}} & \multicolumn{1}{l|}{~0\qquad\quad ~~0\qquad\quad ~1\qquad\quad ~1}                                                                                                                              & \multicolumn{1}{c|}{\multirow{-3}{*}{2}}  & \multicolumn{1}{c|}{\multirow{-3}{*}{21}} \\ \hline
    \multicolumn{1}{|c|}{}                    & \multicolumn{1}{c|}{}                                  & \multicolumn{1}{l|}{NULL}                                                                                                                                                                        & \multicolumn{1}{c|}{}                     & \multicolumn{1}{c|}{}                     \\ \cline{2-3}
    \multicolumn{1}{|c|}{}                    & \multicolumn{1}{c|}{\textit{cost}}    & \multicolumn{1}{l|}{}                                                                                                                                                                                                 & \multicolumn{1}{c|}{}                     & \multicolumn{1}{c|}{}                     \\ \cline{2-3}
    \multicolumn{1}{|c|}{\multirow{-3}{*}{8}} & \multicolumn{1}{c|}{\textit{counter}} & \multicolumn{1}{l|}{}                                                                                                                                                                                                 & \multicolumn{1}{c|}{\multirow{-3}{*}{5}}  & \multicolumn{1}{c|}{\multirow{-3}{*}{30}} \\ \hline
    \end{tabular}
\end{scriptsize}
    \caption{Simulation of the entire process of the example shown in Figure~\ref{fig:4}.} \label{table:1}
    \setlength\tabcolsep{15pt}
\end{table}

\begin{remark}
The reader may wonder whether the costs of the options in $K$ is monotonic (increase or decrease).
  If this were true, our algorithm can be simplified.
However, Table~\ref{table:1} shows that the answer is to the opposite.
 When $i=7$, there are two options in each of $\Kp$ and $\Kip$, so the costs of $K$ is not monotonic.
\end{remark}

\subsection{Extension to a multi-agent version of the problem}\label{subsect:extension}

In this subsection, we discuss an extension that not only partitions the subsequence but also assigns the parts to different agents.

\begin{problem}[Partition and assign problem]\label{problem:extended}
  Given $k$ threshold values $w_{0}^{(1)},\ldots, w_{0}^{(k)}$ together with
   $k$ coefficients $c_1,\ldots,c_k$.
  We have $n$ jobs $1,\ldots,n$ to process (in order), where job $i$ is associated with $(w_i,s_i)$.
	All parameters are nonnegative.
   A group of consecutive jobs $i,\ldots,j$ can be processed in a batch as follows:
    if $w_i+\ldots+w_j\leq w_{0}^{(a)}$ for some $a\in\{1,\ldots,k\}$, jobs $i,\ldots,j$ can be processed in a batch by an agent of type $a$,
       and the cost is $c_a \cdot \max\{s_i,\ldots,s_j\}$.
    Find a partition and assign an agent for each part that minimizes the total cost.
\end{problem}

Comparing to the original problem, we now have $k$ choices for each part.

\smallskip Gladly, our technique shown in the last subsections can be generalized to solving the extended problem.
Let $F[i]$ be the same as before. We have

\begin{equation}\label{eq:extended-dp}
	F[i]=\min \left\{
             \begin{array}{l}
               F_1[i]:=\min_{j:0\leq j<i} \{ F[j]+ c_1 \cdot S_{j+1,i} \mid W_{j+1,i}\leq w_{0}^{(1)}\} \\
					 \ldots\\
               F_k[i]:=\min_{j:0\leq j<i} \{ F[j]+ c_k \cdot S_{j+1,i} \mid W_{j+1,i}\leq w_{0}^{(k)}\}
             \end{array}
           \right..
 \end{equation}

The difficulty lies in computing $F_a[i]$ for $1\leq a \leq k$.

Denote $O^{(a)}_i=\{j \mid 0\leq j <i, W_{j+1,i}\leq w_0^{(a)}\}$ and $o^{(a)}_i=\min{O^{(a)}_i}$.
       Call each element $j$ in $O^{(a)}_i$ an \emph{$a$-option} of $i$.
An $a$-option $j$ of $i$ is regarded as \emph{s-maximal} if $j>0$ and $s_j> S_{j+1,i}$.
  Denote by $\Os_i^{(a)}$ the set of s-maximal $a$-options of $i$.

The following lemma is similar to Lemma~\ref{lem:1}; proof omitted.
\begin{lemma} \label{lem:1-extended}
Set $\Os^{(a)}_i \cup \{o^{(a)}_i\}$ contains an optimal option of $F_a[i]$. As a corollary:
    \begin{equation}\label{eq:4-extended}
    	F_a[i]=\min_j \left\{F[j]+c_a \cdot S_{j+1,i} \mid j \in \Os_i^{(a)} \cup \{o^{(a)}_i\} \right\}.
    \end{equation}
\end{lemma}

The difficulty lies in computing the right part of \eqref{eq:4-extended}.
We can maintain $J^{(a)}=\Os_i^{(a)}$ and find the best $j\in J^{(a)}$ in $O(\log n)$ time using a min-heap.
Or, we can partition $J^{(a)}$ into patient and impatient options as we did for $J$, and find the optimal option in each group in $O(1)$ time using a monotonic queue / stack.
Therefore, we can compute $F_a[i]$ in $O(1)$ amortized time. As a corollary,

\begin{theorem}
Problem~\ref{problem:extended} can be solved in $O(nk)$ time.
\end{theorem}

Moreover, our technique easily extends to solving the following
generalization of problem~\ref{problem:extended}
  (which is suggested by an anonymous reviewer).

\begin{problem}[Partition and assign problem (generalized)]\label{problem:extended-generalized}
  Given $k$ threshold values $w_{0}^{(1)},\ldots, w_{0}^{(k)}$.
  We have $n$ jobs $1,\ldots,n$ to process (in order),
     where job $i$ is described by $(w_i,s_{i,1},\ldots,s_{i,k})$.
	All parameters are nonnegative.
   A group of consecutive jobs $i,\ldots,j$ can be processed in a batch as follows:
    if $w_i+\ldots+w_j\leq w_{0}^{(a)}$ for some $a\in\{1,\ldots,k\}$, jobs $i,\ldots,j$ can be processed in a batch by an agent of type $a$,
       and the cost is $\max\{s_{i,a},\ldots,s_{j,a}\}$.
    Find a partition and assign an agent for each part that minimizes the total cost.
\end{problem}

Note that problem~\ref{problem:extended} is a special case where
   $s_{i,a}=c_a \cdot s_i$.

\begin{theorem}
Problem~\ref{problem:extended-generalized} can also be solved in $O(nk)$ time.
\end{theorem}

\section{Tree partition}\label{sect:tree}

In this section, we move on to the tree partition problem defined as follows.

\begin{problem}\label{problem:tree-partition}
Given two reals $w_0, b$ and a tree. Each vertex $v$ of tree is associated with two real parameters $w_v$ and $s_v$.
  Determine whether the tree (vertices) can be partitioned into several connected components $\{T_k\}$ such that,
\begin{equation}\label{eq:tree-partition}
	\sum_{v \in T_k} w_v \leq w_0~(\forall k) \quad and \quad \sum_{k} \max(s_v\mid v \in T_k)\leq b
\end{equation}
\end{problem}

\newcommand{\NP}{\mathcal{NP}}
\newcommand{\NPC}{\mathcal{NPC}}

Our first result about this problem is a hardness result:
\begin{theorem}\label{thm:npc}
Problem \ref{problem:tree-partition} belongs to $\NPC$, i.e., it is NP-complete.
\end{theorem}

\subsection{A proof of the hardness result}

\begin{problem}\label{problem:knap-sack}
Given a sequence of real numbers $(w_1,\cdots,w_n,s_1,\cdots,s_n,w_0,s_0)$, where $w_i\geq 0~(1\leq i\leq n)$, determine whether there exists a set $A \subseteq [1,n]$ such that
\begin{equation}\label{eq:13}
	\sum_{i \in A} w_i \leq w_0 \quad \text{and} \quad \sum_{i \in A} s_i \geq s_0
\end{equation}
\end{problem}

\begin{problem}\label{problem:knapsack-2}
Given a sequence of real numbers $(w_1,\cdots,w_n,s_1,\cdots,s_n,w_0,s_0)$, where $w_i\geq 0~(1\leq i\leq n)$, determine whether there exists a set $A \subseteq [1,n]$ such that
\begin{equation}\label{eq:14}
	\sum_{i \in A} w_i \leq w_0 \quad \text{and} \quad \sum_{i \in A} s_i - \max_{i \in A} s_i \geq s_0
\end{equation}
\end{problem}

\begin{lemma}\label{lemma:npc}
Problem \ref{problem:knapsack-2} belongs to $\NPC$.
\end{lemma}

\begin{proof}
We will prove that problem~\ref{problem:knap-sack} reduces to problem~\ref{problem:knapsack-2}.
  Further since problem~\ref{problem:knap-sack} $\in \NPC$ (which is well-known \cite{book-IntroAlg}), we obtain that problem~\ref{problem:knapsack-2} $\in \NPC$.

Assume $I = (w_1,\cdots,w_n,s_1,\cdots,s_n,w_0,s_0)$ is an instance of problem~\ref{problem:knap-sack}.
Let $I'= (w_1,\cdots,w_n,w_{n+1}=0,s_1,\cdots,s_n,s_{n+1}=\max\{s_1,\cdots,s_n\},w_0,s_0)$,
  which is an instance of problem \ref{problem:knapsack-2}.
Denote by $\mathbf{L},\mathbf{L}'$ the set of yes instances of problem~\ref{problem:knap-sack}, \ref{problem:knapsack-2} respectively.
It reduces to proving that $I \in \mathbf{L} \Leftrightarrow I' \in \mathbf{L}'$.

Assume that $I \in \mathbf{L}$. This means that there exists $A \subseteq [1,n]$ such that \eqref{eq:13} holds.
It is easy to see that $A \cup \{n + 1\}$ satisfies \eqref{eq:14}, therefore $I' \in \mathbf{L}'$.

Assume that $I' \in \mathbf{L}'$. This means that there exists $A \subseteq [1,n+1]$ such that \eqref{eq:14} holds.
Without loss of generality, assume $n+1\in A$; otherwise $A \cup \{n+1\}$ still satisfies \eqref{eq:14}.
It is easy to see that $A-\{n+1\}$ satisfies \eqref{eq:13}, therefore $I \in \mathbf{L}$.
\end{proof}

With the above lemma, we can now prove Theorem~\ref{thm:npc}.

\begin{proof}[Proof of Theorem~\ref{thm:npc}]
We will show that problem~\ref{problem:knapsack-2} reduces to problem \ref{problem:tree-partition}.
 Further since problem~\ref{problem:knapsack-2} $\in \NPC$ (see Lemma~\ref{lemma:npc}), we obtain that problem~\ref{problem:tree-partition} $\in \NPC$.

Consider an instance of problem~\ref{problem:knapsack-2}, $I = (w_1,\cdots,w_n,s_1,\cdots,s_n,w_0,s_0)$.
Without loss of generality, we assume that each $w_i$ is at most $w_0$.
  Otherwise, we can simply remove $(w_i,s_i)$ from the instance and the answer does not change.

 Let $b = (\sum_{i=1}^{n} s_i) - s_0$. Then, formula \eqref{eq:14} can be rewritten as follows.

\begin{equation}\label{eq:15}
	\sum_{i \in A} w_i \leq w_0 \quad \text{and} \quad \sum_{i \in A} s_i - \max_{i \in A} s_i \geq \left(\sum_{i=1}^{n} s_i\right) - b
\end{equation}
Equivalently,
\begin{equation}\label{eq:16}
	\sum_{i \in A} w_i \leq w_0 \quad \text{and} \quad \max_{i \in A} s_i + \sum_{i \notin A} s_i \leq b.
\end{equation}

Now, we construct an instance $I'$ of problem~\ref{problem:tree-partition} from $I$.
First, build a tree with vertices $1,\ldots,n$ and $n+1$, where $1,\ldots,n$  are all connected to $n+1$.
  The $i$-th~$(1\leq i\leq n)$ node is associated with $w_i$ and $s_i$.
   Moreover, set $w_{n+1}=s_{n+1}=0$.

Note that a partition of this tree corresponds to a subset $A$ of $[1,n]$ --- $A$ contains the labels of those vertices in the same connected component with $n+1$.
Moreover, the cost of the partition $\sum_{k} \max(s_i \mid i\in T_k)$ is $\max_{i \in A} s_i + \sum_{i \notin A} s_i$.
Therefore, subset $A$ satisfies formula \eqref{eq:16}
   if and only if the corresponding partition of $A$ satisfies formula \eqref{eq:tree-partition}.
It follows that $I$ is a yes instance of problem \ref{problem:knapsack-2} if and only if
   $I'$ is a yes instance to problem~\ref{problem:tree-partition}.
Hence the reduction works.
\end{proof}

\subsection{A dynamic programming approach for the case of unit or integer weight}

This subsection considers the tree partition problem under the restriction that
  all the nodes have a unit weight, or more generally, all the nodes have integer weights.
   A dynamic programming approach is proposed, which takes $O(w_0n^2)$ time and $O(w_0^2n^2)$ time, respectively,  for the mentioned unit and integer weight case.
   As we will see, the analysis that it takes $O(w_0n^2)$ time rather than $O(w_0^2 n^2)$
for the unit weight case is nontrivial and is based on interesting observations.

\smallskip Below we mainly focus on the case of unit weight, i.e., assume $w_i$'s are all 1.
It can be easily seen that our approach extends to the integer weight case.

Denote the given tree by $T$, and denote by $T_v$ the subtree rooted at vertex $v$.

\begin{definition}
See Figure~\ref{fig:11}. In any partition of $T_v$, the component containing $v$ is called the \emph{growing component},
  and the other components are called \emph{grown components}. (Within this subsection, a \emph{component} is short for a connected component of $T_v$.)  The \emph{grown part} refers to the set of all grown components.
\end{definition}

For a vertex $v$ and integers $j~(1 \leq j \leq w_0)$ and $k~(1 \leq k \leq n)$,
   let $f[v][j][k]$ be the minimum cost of grown part,
      among all the partitions of $T_v$ whose growing component
        has exactly $j$ nodes and has no $v'$ with $s_{v'}>s_k$. Formally,
\begin{equation}\label{eq:9}
f[v][j][k] = \min_{\begin{subarray}{c}
                \Pi: ~\text{partition of }T_v\text{ with $j$ nodes}\\
                \text{in the growing component, and}\\
                \text{$s_{v'}\leq s_k$ for each such node $v'$.}\end{subarray}} \text{the cost of the grown part of }\Pi.
\end{equation}

To be clear, the cost of the grown part is the total costs of the grown components.
  Moreover, we define $f[v][j][k]= \infty$ in case there is no such partition.

\smallskip Let $F[v]$ be the cost of the optimal partition of $T_v$. Clearly,

\begin{equation}\label{eq:10}
F[v] = \min_{j,k}\{ f[v][j][k] + s_k\}
\end{equation}

\begin{figure}[h]
\begin{minipage}[b]{.55\textwidth}
  \centering
  \includegraphics[width=\textwidth]{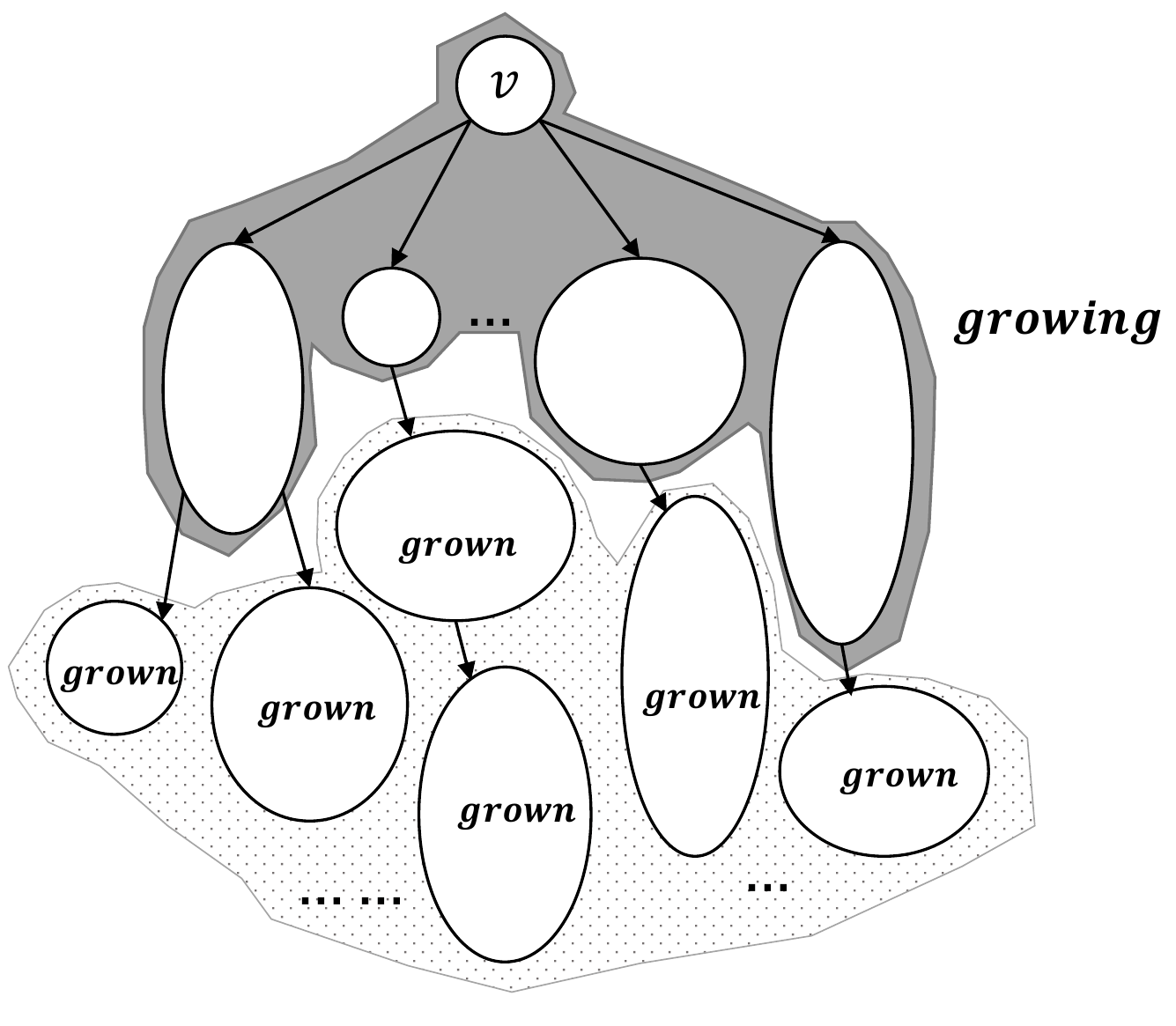}\\
  \caption{Illustration of growing and grown components.}\label{fig:11}
\end{minipage}
\begin{minipage}[b]{.45\textwidth}
  \centering
  \includegraphics[width=.9\textwidth]{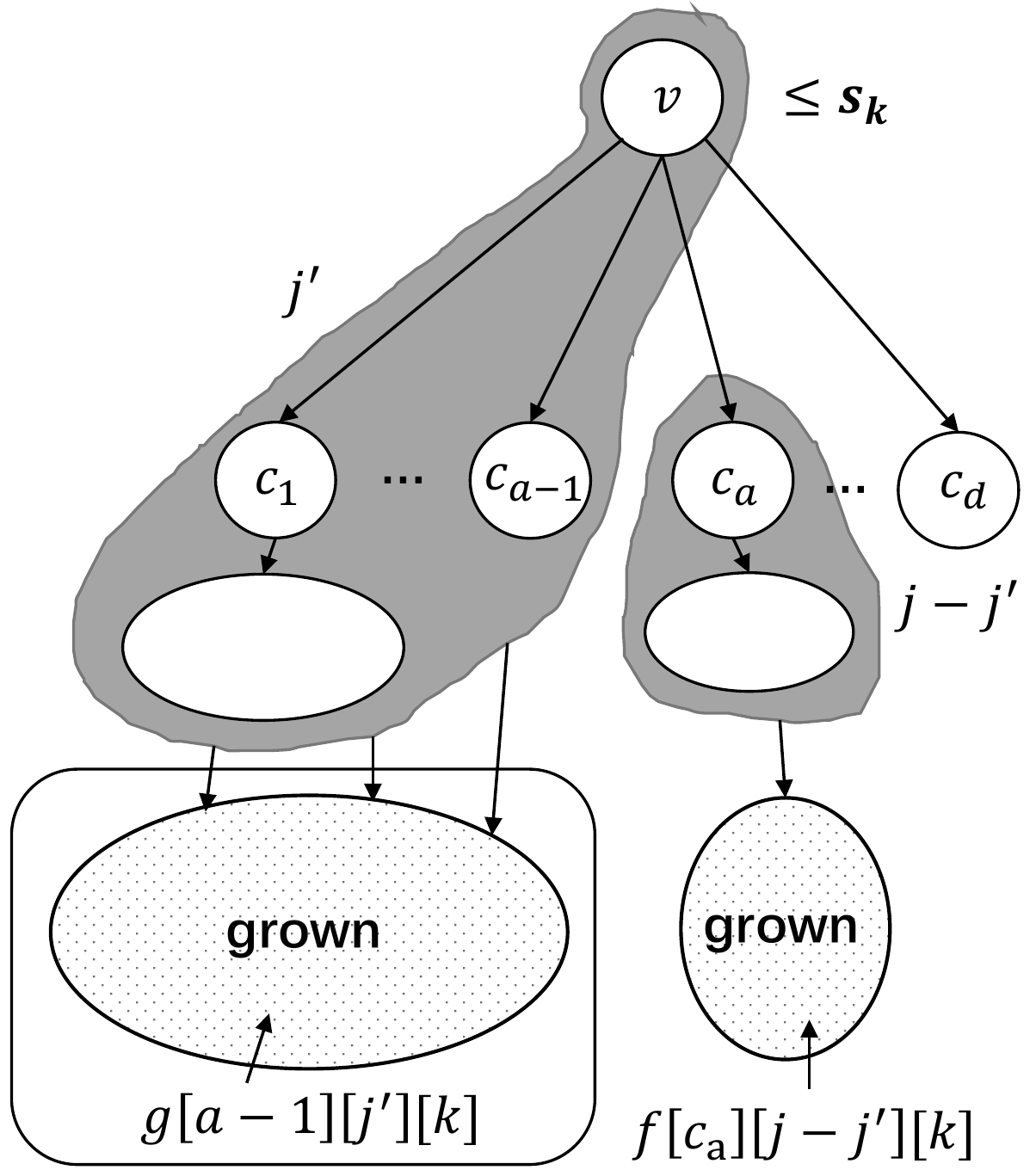}\\
  \caption{Illustration of the computation of $g[a][j][k]$.}\label{fig:12}
\end{minipage}
\end{figure}

We address the computation of $f[v][j][k]$ in the following.

Fix $v$. Assume $v$ has $d$ children $c_1,\ldots, c_d$ (left to right).
Denote by $T_v^a~(0 \leq a\leq d)$ the tree obtained by deleting $c_{a+1}, \cdots, c_d$ and all their descendants from $T_v$.
Let $g[a][j][k]$ be the minimum cost of the grown part,
    among all partitions of $T_v^a$ whose growing component has $j$ vertices and has no $v'$ with $s_{v'}>s_k$.
       To be clear, $g[a][j][k]=\infty$ if no such partition exists. Note that $T_v=T_v^d$ and
\begin{equation} \label{eq:11}
    f[v][j][k]=g[d][j][k].
\end{equation}

It reduces to computing $g[a][j][k]~~(0\leq a\leq d, 1\leq j\leq w_0, 1\leq k \leq n)$.

Assume that $s_v\leq s_k$ and $a>0$. Otherwise it is trivial to get $g[a][j][k]$.

Now, note that $d>0$ (as $a>0$) and therefore $v$ is not a leaf. We have
\begin{eqnarray}\label{eq:12}
        g[a][j][k]&=&\min_{1 \leq j'\leq j} \{g[a-1][j'][k] + \Delta_{j-j'} \},\\
        \text{where }\Delta_{x}&=&\left\{
                      \begin{array}{ll}
                        f[c_a][x][k], & x>0; \\
                        F[c_a], & x=0.
                      \end{array}
                    \right.
\end{eqnarray}

See Figure~\ref{fig:12} for an illustration of \eqref{eq:12}. We omit the easy proof of \eqref{eq:12}.

\subsubsection{Running time analysis of the above algorithm}

    Let $d_v$ be the number of children of $v$.
    It takes $O(j)$ time for computing $g[a][j][k]$ based on \eqref{eq:12},
        so computing the $g$'s take $O(\sum_v d_v w_0^2 n)=O(w_0^2 n^2)$ time.
    It is easy to compute $f$ using \eqref{eq:11} and $F$ using \eqref{eq:10},
        within $O(w_0 n^2)$ time.
    So, the total time is $O(w_0^2 n^2)$. (Be aware that $w_0\leq n$.)

Clearly, the analysis also holds for the integer weight case. Therefore,

\begin{theorem}
When all the nodes have integer weights, the tree partition problem can be solved in $O(w_0^2n^2)$ time.
\end{theorem}

In the following, with a much more careful running time analysis,
  we show that the above dynamic programming approach for the unit weight case costs only $O(w_0n^2)$ time, rather than $O(w_0^2n^2)$ time as what it looks like.

\begin{theorem}\label{thm:running_time_cubic}
When all the nodes have a unit weight, the tree partition problem can be solved in $O(w_0n^2)$ time.
\end{theorem}

The following lemma is crucial for the proof of Theorem~\ref{thm:running_time_cubic} below.

\begin{lemma}\label{lemma:C}
Consider any binary tree $T$ and an integer $K\geq 1$. Let $V$ denote the set of nodes of $T$.
For each $v\in V$, denote by $L_v$ and $R_v$ the number of nodes in the left subtree of $v$
   and the number of nodes in the right subtree of $v$, respectively. We have
\begin{equation}\label{eq:lemmaC-0}
\sum_{v\in V} \min\left(L_v,K\right) \cdot \min\left(R_v,K\right) = O(|V|\cdot K).
\end{equation}
\end{lemma}

\begin{proof}
The following equations together imply \eqref{eq:lemmaC-0}.
\begin{eqnarray}
\sum_{v\in V,L_v<K,R_v<K} L_v \cdot R_v &=& O(|V|\cdot K). \label{eq:lemmaC-1}\\
\sum_{v\in V,L_v <K,R_v\geq K} L_v \cdot K &=& O(|V|\cdot K). \label{eq:lemmaC-2}\\
\sum_{v\in V,L_v\geq K,L_v<K} K\cdot R_v &=& O(|V|\cdot K). \label{eq:lemmaC-3}\\
\sum_{v\in V,L_v\geq K,R_v\geq K} K\cdot K &=& O(|V|\cdot K). \label{eq:lemmaC-4}
\end{eqnarray}

\noindent \emph{Proof of \eqref{eq:lemmaC-1}.}
Call a node $v$ \emph{small} if $L_v$ and $R_v$ are both smaller than $K$.
Denote by $v_1,\ldots,v_s$ all those nodes that are small themselves and their parents are not small.
Moreover, denote the subtree rooted at $v_i$ by $T_{v_i}$ as before.
Clearly,
$$|T_{v_i}|<L_{v_i}+R_{v_i}+1<2K \text{ and } \sum_{v\in T_{v_i}} L_v\cdot R_v \leq |T_{v_i}|^2.$$

It follows that
$$\sum_{v\in V,L_v<K,R_v<K} L_v\cdot R_v
  = \sum_{i=1}^{s} \sum_{v\in T_{v_i}} L_v \cdot R_v
  =\sum_{i=1}^{s} |T_{v_i}|^2
  \leq \sum_{i=1}^{s} 2K\cdot |T_{v_i}|
   \leq 2K \cdot |T|.
$$

The last step is due to the observation that $T_{v_1},\ldots,T_{v_s}$ are disjoint.

\medskip \noindent \emph{Proof of \eqref{eq:lemmaC-2}.}
It reduces to proving that
    $\Sigma_{v\in V,L_v <K,R_v\geq K} L_v = O(|V|)$,
  which further reduces to showing that
      (i) there are $O(|V|)$ vertex pairs $(v,u)$ satisfying the constraint that
            $L_v <K,R_v\geq K$, and $u$ is in the left subtree of $v$.

Fact (i) follows from the observation that for each $u$, there is at most one pair
  $(v,u)$ satisfying the mentioned constraint. See the left picture of Figure~\ref{fig:lemmaC}
   for an illustration. Suppose to the opposite that $(v,u),(v',u)$ are two such pairs, where $v'$ is an ancestor of $v$.  We have $L_{v'}\geq K$ because $R_v\geq K$, contradictory.

\medskip \noindent \emph{Proof of \eqref{eq:lemmaC-3}.}
   This is symmetric to the proof of \eqref{eq:lemmaC-2}.

\begin{figure}[h]
\begin{minipage}[b]{.35\textwidth}
  \centering  \includegraphics[width=.5\textwidth]{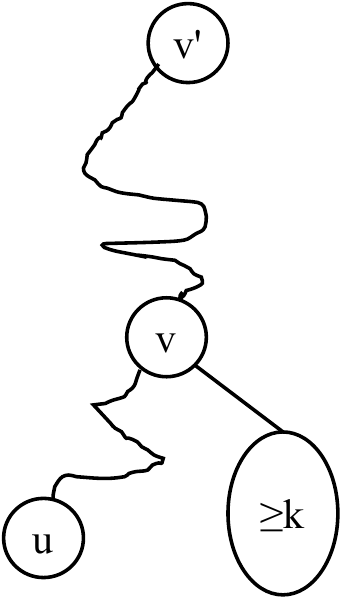}\\
\end{minipage}
\begin{minipage}[b]{.63\textwidth}
  \centering  \includegraphics[width=\textwidth]{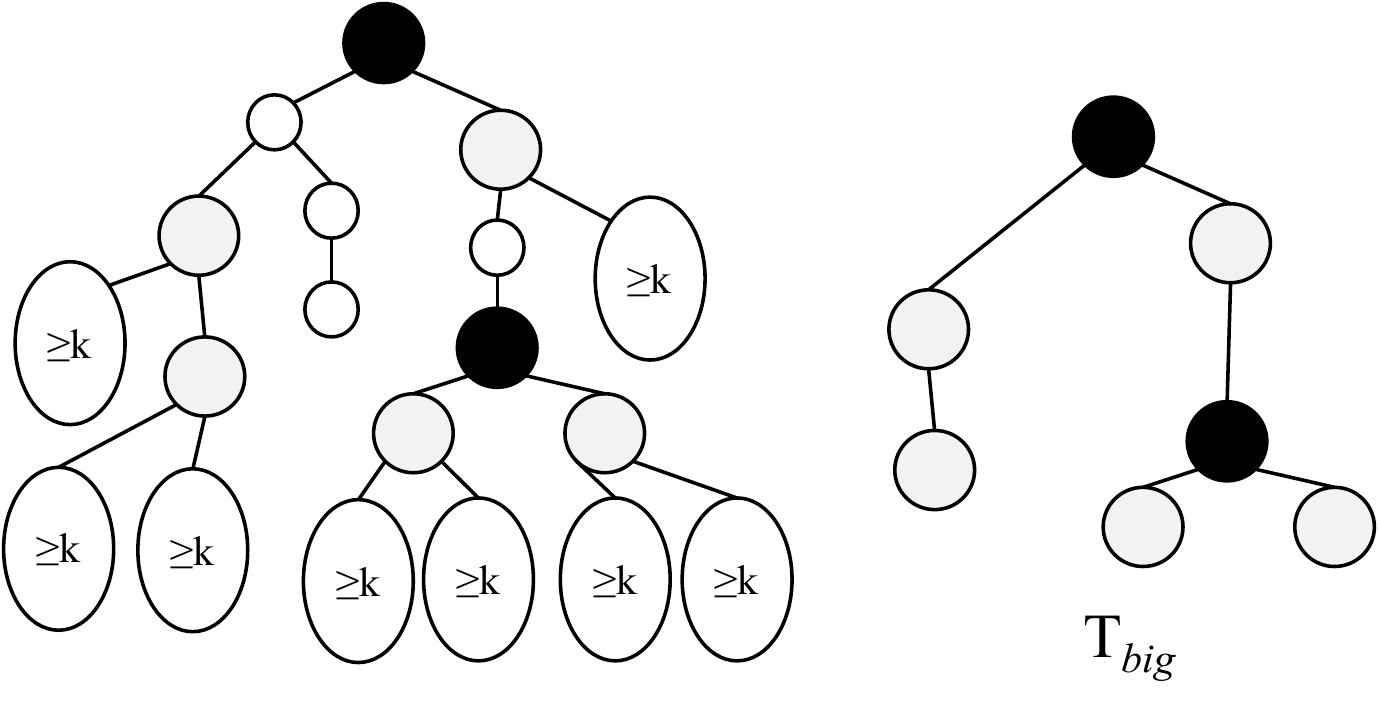}\\
\end{minipage}
  \caption{Illustration of the proof of \eqref{eq:lemmaC-2} and \eqref{eq:lemmaC-4}}
   \label{fig:lemmaC}
\end{figure}

\medskip \noindent \emph{Proof of \eqref{eq:lemmaC-4}.}
   Call a node $v$ \emph{big} if both $L_v$ and $R_v$ are at least $K$.
   The set of big nodes form a tree $T_{big}$, as shown in the right picture of Figure~\ref{fig:lemmaC}.
   Those big nodes with degree 0 or 1 in $T_{big}$ are colored gray in the figure
       and those with degree 2 in $T_{big}$ are colored black.
   Because the nodes with degree 2 are less than the nodes with degree 0,
       the number of big nodes is less than 2 times the number of gray nodes.
   Clearly, to prove \eqref{eq:lemmaC-4},
     it reduces to proving that the number of gray nodes is  $O(|V|/K)$.
	On the other hand, from the picture we can see each gray node is associated with at least $k$ nodes in $T$
     whereas each node is associated to at most one gray node, therefore the number of gray nodes is at most $|V| / (K+1)$.
\end{proof}

We now come back to the proof of Theorem~\ref{thm:running_time_cubic}.

\begin{proof}[Proof of Theorem~\ref{thm:running_time_cubic}]
The bottleneck of the algorithm lies in computing array $g$.
For the $a$-th branch of vertex $v$ $(1<a\leq d_v)$,
  and for $1\leq j\leq w_0, 1\leq k \leq n$,
  we shall compute $g[a][j][k]=min_{1 \leq j'\leq j} \{g[a-1][j'][k] + \Delta_{j-j'} \}$,
where $\Delta_{j-j'}$ can be obtained in $O(1)$ time.
Notice that $j'\leq \min(|T_v^{a-1}|,w_0)$ and $j-j'\leq \min(|T_{c_a}|,w_0)$
   (recall $c_a$ denotes the $a$-th son of $v$ and $T_{c_a}$ denotes the $a$-th subtree of $v$). Therefore, computing $\{g[?][?][k]\}$ for any fixed $k$ requires time proportional to
$$S:=\sum_v \sum_{1<a\leq d_v} \min(|T_v^{a-1}|,w_0) \cdot \min(|T_{c_a}|,w_0).$$
Hence it reduces to showing that $S=O(nw_0).$

Next, convert the given tree $T$ into a binary tree $T'$.
Each node $v$ in $T$ corresponds to a node $v'$ in $T'$.
  The edges in $T'$ are connected as follows.
If the rightmost child of $v$ is $c$,  the right child of $v'$ is $c'$.
If the brother to the left of $v$ is $b$, the left child of $v'$ is $b'$.
Applying this conversion between $T$ and $T'$, we have

\begin{align*}
 S=&\sum_v \sum_{1<a\leq d_v} \min(|L_{c'_a}+1|,w_0) \cdot \min(|R_{c'_a}+1|,w_0)\\
 \leq & \sum_{v'\in V(T')}  \min(L_{v'}+1,w_0)\cdot  \min(R_{v'}+1,w_0)\\
 \leq & \sum_{v'\in V(T')}  \bigg(\min(L_{v'},w_0)+1\bigg)\cdot
			 \bigg(\min(R_{v'},w_0)+1\bigg)\\
 \leq & \sum_{v'\in V(T')} \bigg( \min(L_{v'},w_0) \cdot \min(R_{v'},w_0)+2w_0+1\bigg)
	= O(w_0 |V(T')|).
\end{align*}

The last equality follows from Lemma~\ref{lemma:C}.
\end{proof}

\begin{remark}
Whether Lemma~\ref{lemma:C} has been reported in the literature is not known to us. We guess that this lemma has been found by pioneers in this field, as it is useful in analyzing the running time of similar algorithms on trees. By the way, we wish to see extensions and more elegant proofs of this lemma in the future.
\end{remark}

\section{Summary}

A linear time algorithm is proposed for the Sum-of-Max sequence partition problem under a Knapsack constraint,
  which arises in cargo delivery, telecommunication, and parallel computation.
The algorithm applies a novel dynamic programming speed-up technique
  that partitions the candidate options into groups, such that the options in each group are FIFO or FILO ---
    hence the selection of the best option becomes easy by using monotonic queues and stacks.
In order to efficiently put the options into correct groups, two points are crucial:
  first, introduce the concept of renew for distinguishing options in different states;
  second, use a counter for each option that stores its renewing times in future.
For completeness, we also study the tree partition problem, but it is NP-complete.

Our dynamic programming speed-up technique is applicable in solving
  some variant problems (such as the partition and assign  problems; see problems~\ref{problem:extended}
  and \ref{problem:extended-generalized} in subsection~\ref{subsect:extension}).
In the future, it worths exploring more applications of our speed-up technique
  that divides candidate options into (FIFO or FILO) groups.

\subparagraph{Acknowledgement}
We thank the anonymous reviewers for
   giving us many constructive suggestions on improving the quality of this paper.

\bibliographystyle{elsarticle-num}
\bibliography{CUT}

\clearpage
\appendix
\section{Experimental results}\label{sect:experiment}

We implement the $O(n\log n)$ time algorithm (shown in subsection~\ref{subsect:alg-heap-nlogn}) and
  the $O(n)$ time algorithm (shown in subsection~\ref{subsect:alg-final-on}) by \emph{C/C++ programs},
and test these programs on several test cases and record their running time.

\paragraph{Test cases} We generate two types of test cases, the \emph{special case} where $s_1 > \dots  >s_n$ and $w_0 = n$,
  and the \emph{general case} where $s_1,\ldots,s_n,w_0$ are random.
The $w_i$'s are all set to 1 in all test cases.
(Under the special case, $J$ contains $\Theta(i)$ options in the iteration for computing $F[i]$.
  The special case is the worst case.)
We selects 46 different values for $n$, ranging from 10-1000000 (see Figure~\ref{fig:5}).
\begin{figure}[h]
  \centering
  \includegraphics[width=1.0\textwidth]{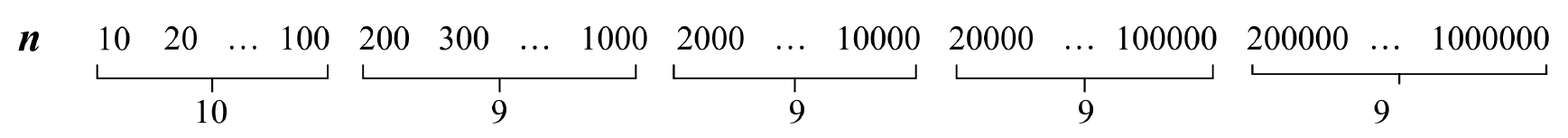}\\
  \caption{The selection of the $n$ value of the number of vertices}\label{fig:5}
\end{figure}

\begin{figure}[h]
\begin{minipage}{0.5\textwidth}
  \centering
  \includegraphics[width=1.0\textwidth]{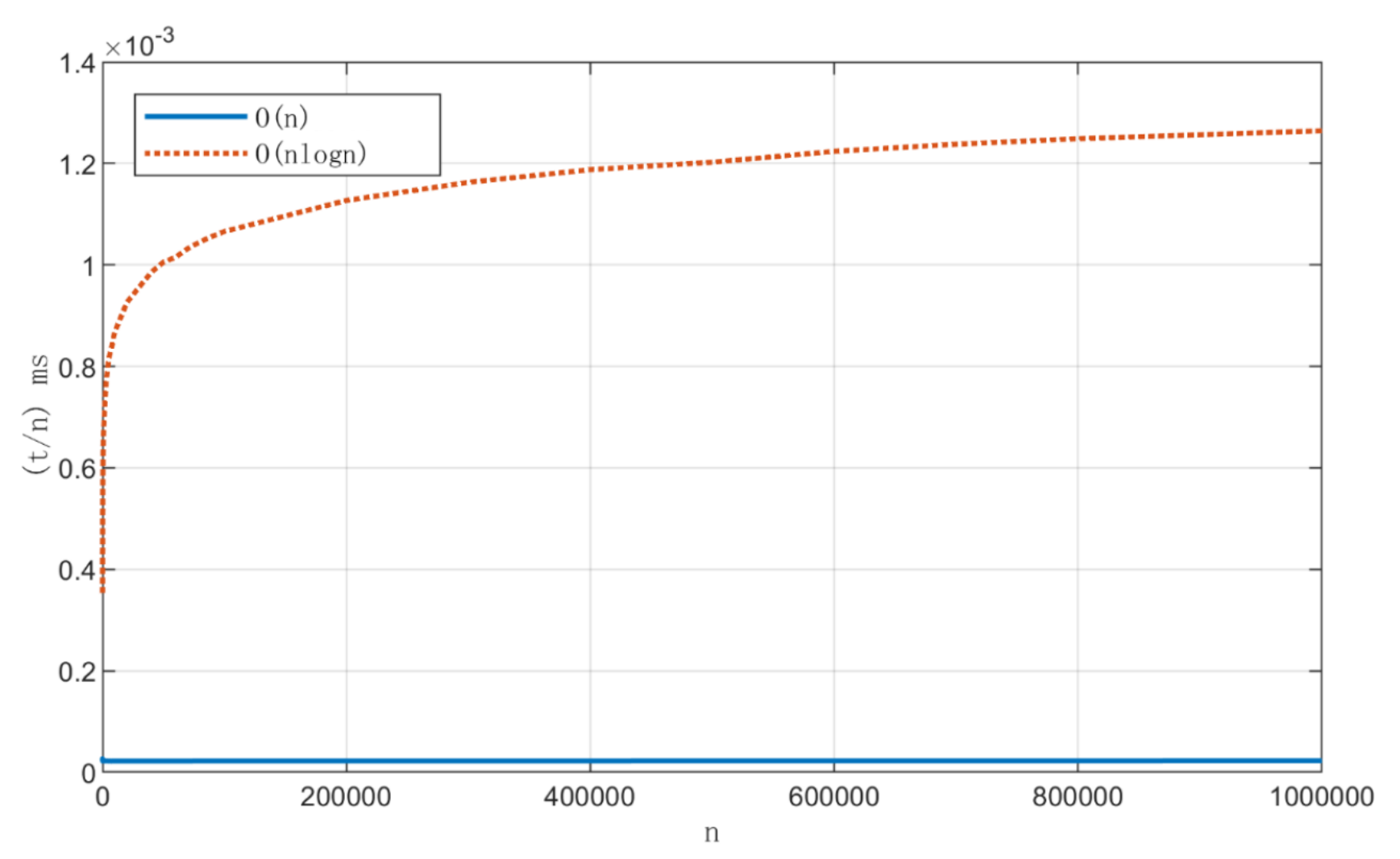}\\
  \caption{Under special case.}\label{fig:7}
\end{minipage}
\begin{minipage}{0.5\textwidth}
  \centering
  \includegraphics[width=1.0\textwidth]{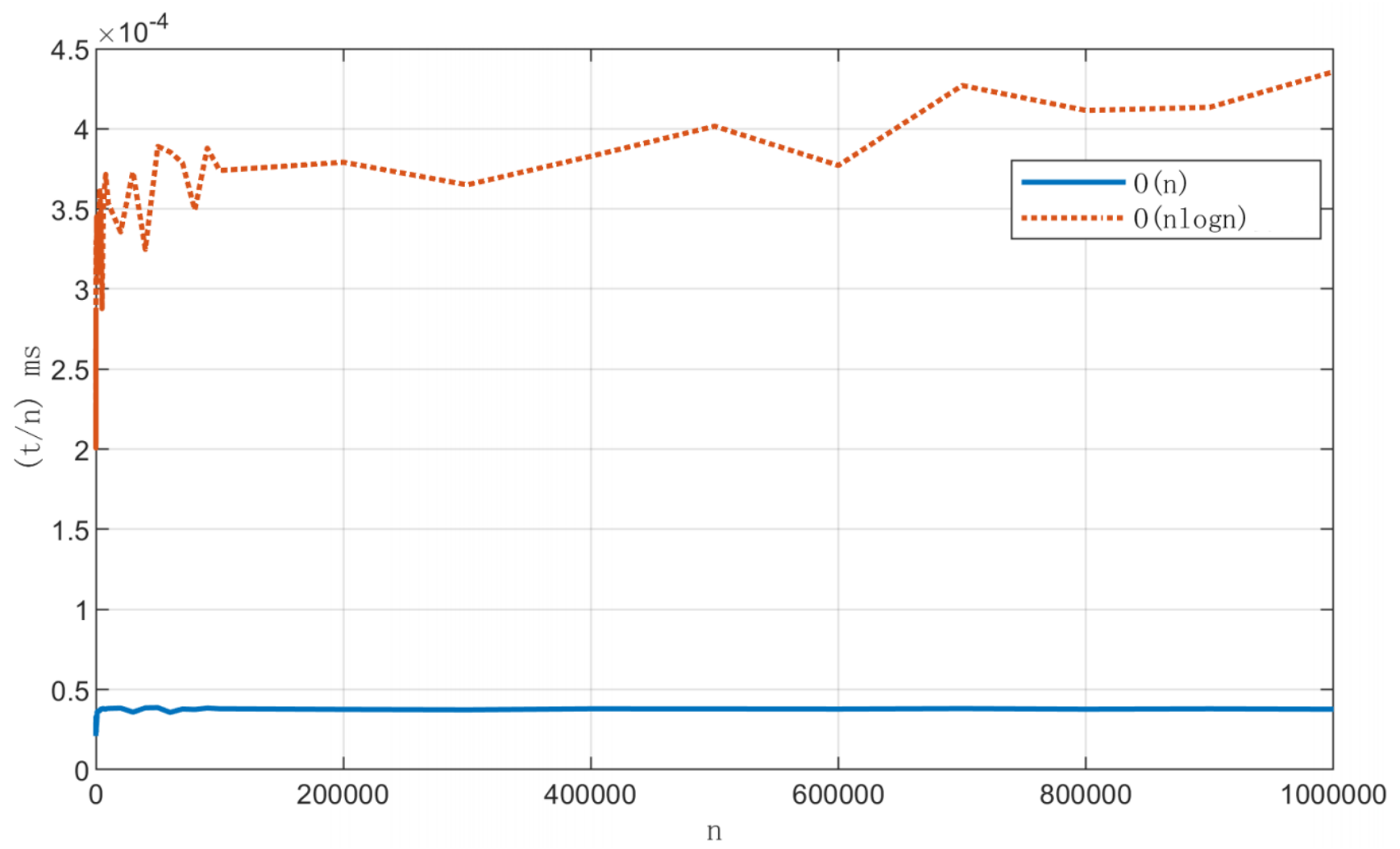}\\
  \caption{Under general case.}\label{fig:8}
\end{minipage}
\end{figure}

Figures~\ref{fig:7} and \ref{fig:8} show the experiment results.
In these graphs, the abscissa indicates the number of vertices $n$, and the ordinate indicates the average of $t/n$, where $t$ represents the running time.
The $t/n$-curve of the $O(n\log n)$ algorithm (orange) grows like a logarithmic function (Figure~\ref{fig:7}), whereas
  the $t/n$-curve of the $O(n)$ algorithm (blue) grows like a constant function.
Therefore, our experimental results are consistent with the analysis of the algorithms.

In both special and general cases, the linear algorithm performs much better.
In particular, it is 60 times faster under the special case when $n=1000000$.

\paragraph{Experiment environment} Operating system: Windows 10.
  CPU: Intel Core i7-10700@2.90GHz 8-core. Memory: 64GB.

\section{The C/C++ code of our linear time algorithm}\label{sect:code}

\begin{small}
\begin{lstlisting}
#include <iostream>
using namespace std;

const int maxn = 100010;
int n, s[maxn], counter[maxn], s_max[maxn], J[maxn], w[maxn];
long long w0, W[maxn];  //W is the prefix sum
long long cost[maxn], f[maxn];

void input_and_preprocess(){
  cin >> n >> w0;
  for (int i = 1; i <= n; i++) cin >> w[i];
  for (int i = 1; i <= n; i++) cin >> s[i];
  int l = 0, r = -1, j = 0;
  for (int i = 1; i <= n; i++){
    W[i] = W[i - 1] + w[i];
    while (W[i] - W[j] > w0) j++;
    while (r >= l && W[i] - W[J[l]] > w0) l++;
    while (r >= l && s[J[r]] <= s[i]) counter[J[r--]]++;
    if (r >= l) counter[J[r]]++;
    J[++r] = i;
    s_max[i] = s[J[l + (J[l] == j)]];
  }
}

void compute_and_output(){
  int l = 0, r = -1, j = 0;
  f[n] = -1;
  for (int i = 1; i <= n; i++){
    while (W[i] - W[j] > w0) j++;
    if (j == i) break; // no answer;
    while (r >= l && W[i] - W[J[l]] > w0) l++;
    while (r >= l && s[J[r]] <= s[i]) r--;
    if (r >= l && cost[J[r]] <= f[J[r]] + s[i]){
      cost[J[r]] = f[J[r]] + s[i];
      counter[J[r]]--;
    }
    if (r>l&&counter[J[r-1]]>0&&cost[J[r-1]]<=cost[J[r]]) r--;
    while (r>l&&counter[J[r]]==0&&cost[J[r-1]]>=cost[J[r]]){
      J[r - 1] = J[r]; r --;
    }
    f[i] = f[j] + s_max[i];
    if (r >= l && cost[J[l]] < f[i]) f[i] = cost[J[l]];
    if (r > l && cost[J[r]] < f[i]) f[i] = cost[J[r]];
    cost[i] = -1; J[++r] = i;
  }
  cout << f[n];
}

int main(){
  input_and_preprocess();
  compute_and_output();
  return 0;
}
\end{lstlisting}
\end{small}

\section{Concave 1-d speed-up technique is not applicable}
Z. Galil and K. Park \cite{DP-concave} considered a 1-d dynamic programming equation of formula \eqref{eq:7}, where $E[j]$ can be computed from $F[j]$ in O(1) time, and they pointed out that there are many applications for formula \eqref{eq:7}, e.g. the minimum weight subsequence problem is a special case of this problem.
\begin{equation}\label{eq:7}
	F[i]=\min_{0 \leq j < i} \{E[j]+v(j,i) \}, 0 \leq i \leq n
\end{equation}

Galil and Park designed an ingenious $O(n)$ time algorithm for solving \eqref{eq:7} under the case where $v(j,i)$ satisfies the following \emph{concave property} (briefly, they reduced the problem to solving several totally-monotone matrix searches).

\begin{definition}\label{def:concave}
 The cost function $v$ is concave if it satisfies the quadrilateral inequality:
 \begin{equation}\label{eq:8}
	v(a,c) + v(b,d) \leq v(b,c) + v(a,d), a \leq b \leq c \leq d
\end{equation}
 \end{definition}

We show in the following that the function $S_{a,b}=\max_i \{s_i \mid a \leq i \leq b \}$ is \textbf{not} concave.
  Therefore, the 1-d concave dynamic programming speed-up technique of Galil and Park is \textbf{not} applicable to our circumstance.

Assume $s_1=3, s_2=1, s_3=1, s_4=3$. Then $S_{1,4} = \max\{3,1,1,3\}=3$, $S_{2,3} = \max \{1,1\}=1$, $S_{1,3}=\max\{3,1,1\}=3$, $S_{2,4}=\max\{1,1,3\}=3$. Clearly, $S_{1,3}+S_{2,4}>S_{2,3}+S_{1,4}$, that is, $S_{a,c}+S_{b ,d}>S_{b,c}+S_{a,d}$ for $a=1,b=2,c=3,d=4$, which means that $S_{a,b}$ does not satisfy \eqref{eq:8} and hence is not concave.

\medskip We also mention that the speed-up technique of \cite{DP-convex} is not applicable.
\clearpage

\begin{figure}[h]
\begin{minipage}[c]{.4\textwidth}
  \includegraphics[width=.8\textwidth]{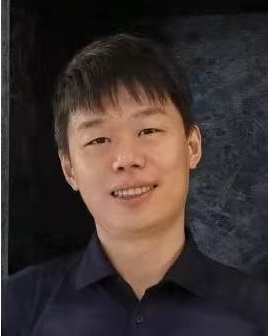}
\end{minipage}
\begin{minipage}[c]{.6\textwidth}
Professor \emph{Kai Jin} was born in Changsha, Hunan, China, in 1986.
He received the B.S. degree and Ph.D. degree in computer science and technology from Tsinghua University, Beijing, China, in 2008 and 2016, respectively.
He was a Postdoc with the HKU from 2016 to 2018
  and with the HKUST from 2018 to 2020.

Since 2020, he joined the School of Intelligent Systems Engineering in Sun Yat-sen University
  as an Associated Professor.
His research area includes algorithm design, combinatorics, game theory, and computational geometry.
\end{minipage}
\begin{minipage}[c]{\textwidth}
~\vspace{45pt}
\end{minipage}
\begin{minipage}[c]{.4\textwidth}
  \includegraphics[width=.8\textwidth]{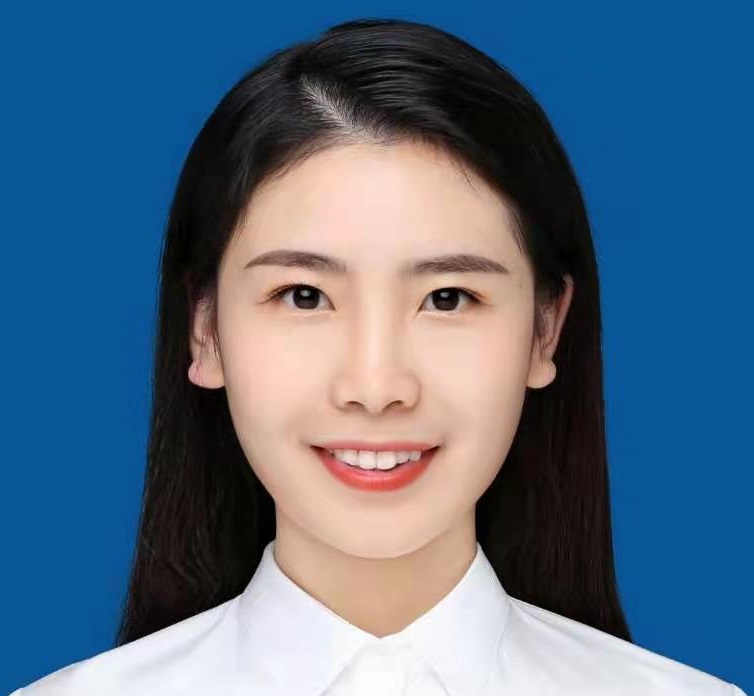}
\end{minipage}
\begin{minipage}[c]{.6\textwidth}
\emph{Danna Zhang} was born in Hengyang, Hunan, China, in 1998. She received the B.S. degree in internet of things from Dalian Maritime University, Dalian, China. She is currently pursuing the M.S. degree in theoretical computer science (supervised by Prof Jin) at Sun Yat-Sen University, Shenzhen, Guangdong, China. Her main research field is algorithm design.
\end{minipage}
\begin{minipage}[c]{\textwidth}
~\vspace{45pt}
\end{minipage}
\begin{minipage}[c]{.4\textwidth}
  \includegraphics[width=.8\textwidth]{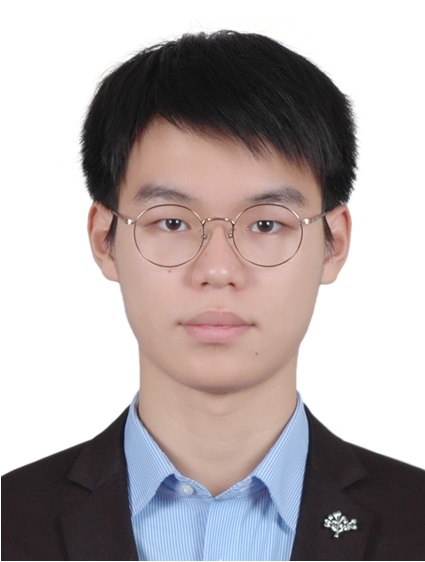}
\end{minipage}
\begin{minipage}[c]{.6\textwidth}
Canhui Zhang was born in JingZhou, Hubei, China, in 2000. He is currently pursuing the B.S. degree in intelligent science and technology with Sun Yat-sen University, Shenzhen, Guangdong, China. His main research interests include algorithm design.
\end{minipage}
\end{figure}

\end{document}